\documentclass[10pt,journal]{IEEEtran}
\usepackage{makecell}
\usepackage[T1]{fontenc}
\usepackage[utf8]{inputenc}

\usepackage{float}
\usepackage{indentfirst}
\usepackage{enumerate}
\usepackage{url}
\usepackage[dvipsnames]{xcolor}
\usepackage[stable]{footmisc}
\usepackage{fancyhdr}
\usepackage{hyperref}
\usepackage[all]{xy}
\usepackage{graphicx}
\usepackage{youngtab}
\usepackage{mathtools}
\usepackage{enumerate}
\usepackage{textcomp}
\usepackage{newcent}
\usepackage[sc,noBBpl]{mathpazo}
\usepackage{amsmath,amsfonts,amssymb,amsthm}
\usepackage[mathscr]{eucal}
\usepackage{soul}
\usepackage{tcolorbox}

\ifCLASSINFOpdf
\else
\fi
\begin{document}
%
\title{Efficient multi port-based teleportation schemes}
%
%
%

\author{Micha{\l} Studzi\'nski, Marek Mozrzymas, Piotr Kopszak and Micha{\l} Horodecki}

\theoremstyle{plain}
\newtheorem{theorem}{Theorem}
\newtheorem{lemma}[theorem]{Lemma}
\newtheorem{proposition}[theorem]{Proposition}
\newtheorem{corollary}[theorem]{Corollary}
\newtheorem{conjecture}[theorem]{Conjecture}
\newtheorem{definition}[theorem]{Definition}
\newtheorem{fact}[theorem]{Fact}
\newtheorem{notation}[theorem]{Notation}
\newtheorem{observation}[theorem]{Observation}
\newtheoremstyle{note}{\topsep}{\topsep}{\slshape}{}{\scshape}{}{ }{}
\theoremstyle{note}
\newtheorem{note}[theorem]{Note}
\newtheorem{remark}[theorem]{Remark}
\newtheorem{example}[theorem]{Example}

\newcommand\const{\operatorname{const}}
\newcommand\vect{\operatorname{Vect}}
\newcommand\vspan{\operatorname{span}}
\newcommand\supp{\operatorname{Supp}}
\newcommand\Div{\operatorname{div}}
\newcommand\rank{\operatorname{rank}}
\newcommand\diag{\operatorname{diag}}
\newcommand\sign{\operatorname{sign}}
\renewcommand\Re{\operatorname{Re}}
\renewcommand\Im{\operatorname{Im}}
\newcommand\im{\operatorname{im}}
\newcommand\tr{\operatorname{Tr}}
\newcommand\res{\operatorname{res}}
\newcommand\result{\operatorname{Res}}
\newcommand\mult{\operatorname{mult}}
\newcommand\spectr{\operatorname{spectr}}
\newcommand\ord{\operatorname{ord}}
\newcommand\grad{\operatorname{grad}}
\newcommand\lt{\operatorname{\mbox{\textsc{lt}}}}
\newcommand\card{\operatorname{card}}
\newcommand\trdeg{\operatorname{tr.\!deg}}
\newcommand\gdeg{\operatorname{\deg_{\vgamma}}}
\newcommand\sdeg{\operatorname{\mbox{$\vsigma$}-deg}}
\newcommand\ldeg{\operatorname{\mbox{$\vlambda$}-deg}}
\newcommand\id{\operatorname{\mathrm{Id}}}
\newcommand\rmd{\mathrm{d}}
\newcommand\SLtwoC{\ensuremath{\mathrm{SL}(2,\C)}}
\newcommand\CPOne{\ensuremath{\C\mathbb{P}^1}}
\newcommand\CP{\ensuremath{\C\bbP}}
\newcommand\VE[1]{\mathrm{VE}_{#1}}
\newcommand\rmi{\mathrm{i}\mspace{1mu}}
\newcommand\rme{\mathrm{e}}
\newcommand\Dt{\frac{\mathrm{d}\phantom{t} }{\mathrm{d}\mspace{1mu} t}}
\newcommand\Dy{\frac{\mathrm{d}\phantom{y} }{\mathrm{d}y}}
\newcommand\Dz{\frac{\mathrm{d}\phantom{z} }{ \mathrm{d}z}}
\newcommand\Dwz{\frac{\rmd w }{\mathrm{d} z}}
\newcommand\Dtt{\frac{\mathrm{d}^2\phantom{t} }{\mathrm{d}t^2}}
\newcommand\pder[2]{\dfrac{\partial #1 }{\partial #2}}
%
%
\newcommand\abs[1]{\lvert #1 \rvert}
\newcommand\norm[1]{\lVert #1 \rVert}
\newcommand\cc[1]{\overline#1}
\newcommand\figref[1]{Fig.~\ref{#1}}
\newcommand\scalar[2]{\langle \mvector{#1},\mvector{#2}\rangle}
\newcommand\pairing[2]{\langle {#1}, {#2}\rangle}
%
%
\newcommand\rscalar[2]{\langle #1, #2\rangle}
%
%
\newcommand\cn{\operatorname{cn}}
\newcommand\sn{\operatorname{sn}}
\newcommand\dn{\operatorname{dn}}
\newcommand\bfi[1]{{\bfseries\itshape{#1}}}
\newcommand\cpoly{\C[x_1,\ldots,x_n]}
\newcommand\crat{\C(x_1,\ldots,x_n)}
\newcommand\mtext[1]{\quad\text{#1}\quad}
\newcommand\mnote[1]{\marginpar{\tiny{#1}}}
\newcommand\defset[2]{\left\{{#1}\;\vert \;\; {#2} \,\right\}}
\newcommand\deftuple[2]{\left({#1}\;\vert \;\; {#2} \,\right)}

\newcommand{\<}{\langle}
\renewcommand{\>}{\rangle}
\newcommand{\ind}{\operatorname{ind}}

\newcommand\hcal{\cH}
\newcommand\be{\begin{equation}}
\newcommand\ee{\end{equation}}
\newcommand\bea{\begin{array}}
	\newcommand\eea{\end{array}}
\newcommand\ben{\begin{eqnarray}}
\newcommand\een{\end{eqnarray}}
\newcommand\ot{\otimes}
\newcommand\tred{\textcolor{red}}
\newcommand\cred{\color{red}}
\newcommand\blk{\color{black}}

\definecolor{forest}{RGB}{11, 102,35}
\newcommand\tforest{\textcolor{forest}}

\newcommand\bei{\begin{itemize}}
	\newcommand\eei{\end{itemize}}
\newcommand\bee{\begin{enumerate}}
	\newcommand\eee{\end{enumerate}}
\DeclarePairedDelimiter\ceil{\lceil}{\rceil}
\DeclarePairedDelimiter\floor{\lfloor}{\rfloor}

\newcommand{\mh}[1]{ \textcolor{blue}{{\tt MH}: #1}}

\maketitle

	\begin{abstract}
In this manuscript we analyse generalised port-based teleportation (PBT) schemes, allowing for transmitting more than one unknown quantum state (or a composite quantum state) in one go, where the state ends up in several ports at Bob's side.  We investigate the efficiency of our scheme discussing both deterministic and probabilistic case, where parties share maximally entangled states. It turns out that the new scheme gives better performance than various variants of the optimal PBT protocol used for the same task. All the results are presented in group-theoretic manner depending on such quantities like dimensions and multiplicities of irreducible representations in the Schur-Weyl duality. 
The presented analysis was possible by considering the algebra of permutation operators acting on $n$ systems distorted by the action of partial transposition acting on more than one subsystem.  Considering its action on the $n-$fold tensor product of the Hilbert space with finite dimension, we present construction of the respective irreducible matrix representations, which are in fact matrix irreducible representations of the Walled Brauer Algebra. I turns out that the introduced formalism, and symmetries beneath it, appears in many aspects of theoretical physics and mathematics - theory of anti ferromagnetism, aspects of gravity theory or in the problem of designing  quantum circuits for special task like for example inverting an unknown unitary.
	\end{abstract}

\begin{IEEEkeywords}
quantum information, quantum teleportation, group representation theory, symmetric group, port-based teleportation.
\end{IEEEkeywords}

\section{(Multi) Port-based teleportation protocols and their importance}
Quantum teleportation is one of the most important primitives in quantum information science. It performs an unknown quantum state transmission between two spatially separated systems. It requires pre-shared entangled resource state and consists of three elements: joint measurement, classical communication and \textit{correction operation} depending on the result of the measurements. 
Except  quantum teleportation protocol presented by Bennett et al. in~\cite{bennett_teleporting_1993} we distinguish Knill-Laflamme-Milburn (KLM) scheme~\cite{knill_scheme_2001}, based solely on linear optical tools and so-called Port-based Teleportation (PBT) protocols, introduced in~\cite{ishizaka_asymptotic_2008}.  Although, standard teleportation and KLM scheme are of the great importance and have fundamental meaning for the field with range of important applications~\cite{boschi_experimental_1998,gottesman_demonstrating_1999,gross_novel_2007, jozsa_introduction_2005, pirandola_advances_2015, raussendorf_one-way_2001, zukowski_event-ready-detectors_1993}, here we focus on PBT schemes. One of the main reasons of that is the PBT is the only scheme where in the last step the \textit{unitary correction is absent}.
The lack of correction in the last step allows for entirely new applications in modern quantum information science and the high amount of its symmetries make it tempting for analysis by representation-theoretic methods. For instance, PBT has found its place in non-local quantum computations and position-based cryptography~\cite{beigi_konig} resulted in new attacks on the cryptographic primitives, reducing the amount of consumable entanglement from doubly exponential to exponential, communication complexity~\cite{buhrman_quantum_2016} connecting the field of communication complexity and a Bell inequality violation,  theory of universal programmable quantum processor performing computation by teleportation~\cite{ishizaka_asymptotic_2008}, universal simulator for qubit channels~\cite{sim} improving simulations of the amplitude damping channel and allowing to obtain limitations of the fundamental nature for quantum channels discrimination~\cite{limit}.
Some aspects of PBT play a role in the general theory of  construction of universal quantum circuit for inverting general unitary operations~\cite{PhysRevLett.123.210502} as well as theory of storage and retrieval of unitary quantum channels~\cite{Stroing}.

In the original formulation of PBT scheme, see Figure~\ref{FPBT}, two parties share a resource state consisting of $N$ copies of maximally entangled state $|\psi^+\>$, each of them called a port. 
\begin{figure}[h]
	\begin{centering}
		\includegraphics[width=0.4\textwidth]{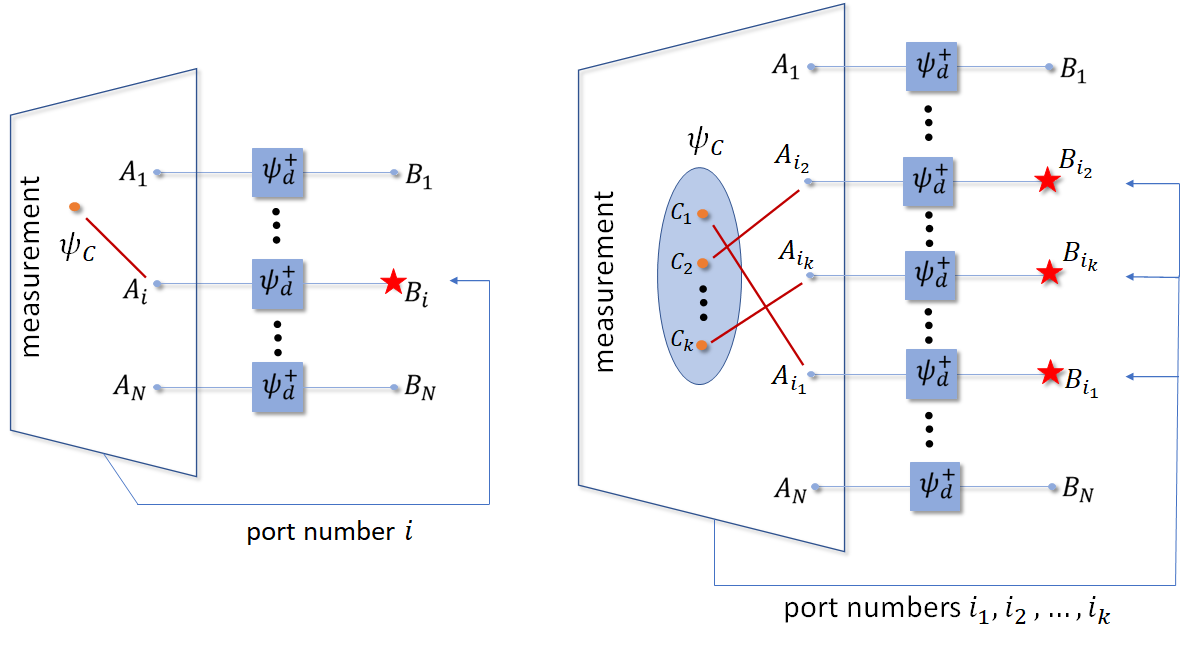}
		\caption{On the left-hand side we present the vanilla scheme for the standard PBT. Two parties share $N$ copies EPR pairs $\Phi_d^+=|\psi^+_d\>\<\psi^+_d|$, where $|\psi^+_d\>=(1/\sqrt{d})\sum_i |ii\>$. Alice to teleport an unknown state $\psi_C$ applies a joint measurement (the blue trapeze) on the state to be teleported and her half $A_1\cdots A_N$ of the resource state, getting a classical outcome $i$ transmitted to by by a classical channel. The index $i$ indicates port on the Bob's side (red star) on which teleported state appears. On the right-hand side we present basic scheme for multi-port teleportation scheme. Again, two parties share $N$ copies EPR pairs $\Phi_d^+=|\psi^+_d\>\<\psi^+_d|$, where $|\psi^+_d\>=(1/\sqrt{d})\sum_i |ii\>$. Alice  to teleport an unknown joint state $\psi_C=\psi_{C_1C_2\cdots C_k}$, where $k\leq \floor{N/2}$, to Bob performs a global measurement (the blue trapeze) on systems $C_1\cdots C_kA_1\cdots A_N$, getting a classical outcome $\mathbf{i}=(i_1,i_2,\ldots,i_k)$. She transmits the outcome $\mathbf{i}$ via classical communication to Bob. The index $\mathbf{i}$ indicates on which $k$ ports on the Bob's side the teleported state  arrives (red stars). Bob to recover  the teleported state has to pick up ports indicated by $\mathbf{i}$ with the right order.} 
		\label{FPBT}%
	\end{centering}
\end{figure}
Alice to teleport an unknown state $\psi_C$ to Bob performs a joint measurement on it and her half of the resource state, communicating the outcome through a classical channel to Bob. It turns out that the outcome received by Bob points to the system in the resource state where the state has been teleported to. We distinguish two versions of PBT protocol - \textit{deterministic} (dPBT) and \textit{probabilistic} (pPBT). In the first case, after the measurement Alice obtains a classical outcome $i\in \{1,\ldots, N\}$. In this scenario, the unknown state is always teleported, but it decoheres during the process. To learn about the efficiency we compute entanglement fidelity, checking how well we are able to transmit half of the maximally entangled state. From the no go theorem~\cite{bennett_teleporting_1993} for the deterministic universal processor, we know that we can achieve perfect teleportation only in the asymptotic limit $N\rightarrow \infty$. In the second case, the probabilistic one, Alice obtains a classical outcome $i\in \{0,1,\ldots,N\}$, where index 0 corresponds to an additional measurement  $\Pi_0^{AC}$ indicating the failure of the teleportation process. In all other cases in pPBT, when $i\in \{1,\ldots,N\}$, parties proceed with the procedure getting teleported state perfectly. To learn about efficiency, we compute the average probability of success of such a process. Similarly, as in the deterministic case, the probability is equal to 1 only in the asymptotic limit  $N\rightarrow \infty$. In every case, we can consider also \textit{optimised PBT}, where Alice optimises jointly over the shared state and measurements before she runs the protocol to increase the efficiency, see~\cite{ishizaka_quantum_2009} for further details. 

Effective evaluation of the performance of both variants of PBT requires determining all symmetries that occur in the problem and spectral analysis of certain operators. For qubits it has been done in~\cite{ishizaka_asymptotic_2008,ishizaka_quantum_2009} by exploiting representation theory of $SU(2)^{\ot N}$, in particular properties of Clebsch-Gordan (CG) coefficients, together with semidefinite programming. Unfortunately, such methods do not work effectively in a higher dimension, $d>2$.  It is because in the case of $SU(d)^{\ot N}$ there is no closed-form of the CG coefficients and to compute them we need an exponential overhead in $N$ and $d$.

The first attempt to describe the efficiency of PBT in higher dimensions has been done in~\cite{wang_higher-dimensional_2016} by exploiting elements of Temperley-Lieb algebra theory, mostly in its graphical representation. The authors presented closed expressions for entanglement fidelity as well as the probability of success for an arbitrary $d$ and $N=2,3,4$. 

Next, in papers~\cite{Studzinski2017,StuNJP,MozJPA}, authors develop new mathematical tools allowing for studies of PBT for arbitrary $N$ and $d$. From a technical point of view, the crucial role is played by the algebra of partially transposed permutation operators and its irreducible components. Or in the other words irreducible representations of the commutant of $U^{\otimes (n-1)}\otimes \overline{U}$, where the bar denotes complex conjugation, and $U$ is an element of unitary group $\mathcal{U}(d)$. It turns out that basic objects describing all variants of PBT belongs to the mentioned commutant. Knowing the full description of irreducible spaces we can reduce the analysis to every block separately and present entanglement fidelity and the probability of success in terms of parameters describing respective irreducible blocks like multiplicity or dimension.
Finally, in paper~\cite{majenz} authors investigated the asymptotic behaviour of PBT schemes which was uncovered in the previous works. Their results required advanced tools coming from connections between representation-theoretic formulas and random matrix theory.

Despite of all the results presented above still, we have many important questions to answer in the field of PBT protocols. 
Here we focus on the following problem:
\textit{What is the most effective way to teleport using PBT-like protocols a state of composite system or several systems, let us say $k$?} One of the answer could be the following:
\begin{itemize}
    \item The most obvious one is to run the original PBT with dimension of the port equal to $d^k$, however the performance of the PBT protocols gets worse with growing local dimension~\cite{majenz,StuNJP}.
    \item We could also keep dimensions of the ports and split the resource state into $k$ packages and then run $k$ separate PBT procedures independently. Such analysis, together with some aspets of asymptotic discussion of the teleportation protocols analysed here is studied in \cite{kopszak2021multiport}.
\end{itemize}
In the next sections of this paper, we show that allowing Bob for a mild correction in a form of ports permutation we can find a class of multi-port teleportation protocols (see the right panel of Figure~\ref{FPBT}), allowing for high performance measured in terms of entanglement fidelity or probability of success. Such class of protocols allows us to transfer the state with higher performance than the respective PBT schemes mentioned above.
To obtain the final answers
we deliver {\it novel mathematical tools} concerning both standard Schur-Weyl duality based on $n$-fold tensor product of unitary transformations, $U^{\ot n}$,
as well as, its "skew" version based on the product of type $U^{\otimes N} \otimes \overline{U}^{\otimes k}$ (where bar denotes complex conjugation). By considering irreducible representations of the commutant of $U^{\otimes N} \otimes \overline{U}^{\otimes k}$  (with $n=N+k$), we show its connection with the \textit{algebra of partially transposed permutation operators} $\mathcal{A}_n^{(k)}(d)$, composed of all linear combinations of the standard permutation operators deformed by the operation of partial transposition over last $k$ subsystems. In fact, our work covers unexplored earlier field of finding irreducible matrix representations of Walled Brauer Algebra~\cite{Cox1}. 

The tool kit presented here is not tailored only for effective description of port-based like teleportation protocols and mentioned  kind of symmetries appear in many problems of modern physics and mathematics. 

From the perspective of physics, studying quantum systems with such symmetries  play an important role in  antiferromagnetic systems~\cite{Candu2011}. In this paper the author considers  the spectrum of an integrable antiferromagnetic Hamiltonian of the $gl(M|N)$ spin chain of alternating fundamental and dual representations. In particular, to reduce the complexity of the numerical diagonalisation of the considered Hamiltonian author applies non-trivial tools emerging from the theory of Walled-Brauer Algebra. Here, our new tools possibly enable more analytical approach to the problem or at least further numerical simplifications.

Similar kinds of symmetries have found their place even in some aspects of gravity theories~\cite{1126-6708-2007-11-078,2008PhRvD..78l6003K} and particle physics~\cite{Kimura2010}. Here, authors by applying elements of the representation theory of (Walled) Brauer Algebras and Schur-Weyl duality focus on diagonalisation of the two-point functions of gauge invariant
multi-matrix operators. In particular, they describe how labels appearing in diagonal bases are related to respective Casimir operators and irreducible components of the Brauer and Walled Brauer Algebra. It turns out that the depper understanding the spectrum of states from the point of view of the conformal field theory (CFT) yields information about space-time physics via the AdS/CFT duality~\cite{Maldacena1999}.

Next, our analysis could be applied in the study of the theory of entanglement and positive maps. The first such approach has been made by Werner and Eggeling in seminal paper~\cite{PhysRevA.63.042111}, where the full analysis of tripartite $U \otimes U \otimes U$ states has been made, concentrating on their positivity after partial transposition property (PPT), which is equivalent to considering $U\otimes U\otimes \overline{U}$ invariant operators. Our tools can in principle be used for the  characterisation of multipartite states after having previously chosen the systems to be transposed. This field, although old, is still under exploration, more in the context of positive and $k-$positive maps. To support our claim let us consider recent papers by Collins and co-authors~\cite{Collins_2018,Bardet_2020}.

Furthermore, motivated by the recent results on Temperley–Lieb Quantum Channels~\cite{Brannan_2020}, one can apply methods developed here, together with the above investigations, to the problem of constructive examples of new quantum channels for which the minimum output R{\'e}nyi entropy is not additive.
 
The tools described in paper are enough for the full description of the universal $M\rightarrow N$ quantum cloning machines (where $M<N$) in the group-theoretic manner. 
Such approach has been successfully for  universal $1\rightarrow N$ quantum cloning machines in~\cite{PhysRevA.89.052322}. 

Finally, the methods developed here are very similar to the techniques used in abstract harmonic analysis for non-commutative groups, where irreducible representations play a crucial role~\cite{harmonic_analysis}. This similarity strongly suggests possibility of implementing our mathematical results to some aspect of harmonic analysis in future.

To address at least a part of described above problems we need to diagonalize and investigate properties of some operators  representing certain physical quantity. To do so we have to construct an analogue of the celebrated Young-Yamanouchi basis for the symmetric group $S(n)$. This is the only way of investigating operators which are $U^{\otimes n}$ invariant, thus having non-trivial component only the symmetric part in the Schur-Weyl duality.   However, in our case our symmetry is deformed - we have $k$ complex conjugations - the straightforward approach suggested by the Schur-Weyl duality is not enough. Also, combining with the approach based on considering the dual representation to $\overline{U}$ as it was done~\cite{majenz} is not enough, since we must have full information about matrix entries of the respective operators on the non-trivial sectors, similarly as it is for $U^{\otimes n}$ invariant operators, and pre-existing methods simply do not have access to these sectors of the space. 

From the perspective of pure mathematics we deliver tools for studying and understanding the Walled Brauer Algebras~\cite{Cox1}, which is a sub-algebra of the  Brauer Algebra~\cite{Brauer1} on the most friendly level for potential applications - irreducible matrix representation. Namely, the algebra of partially transposed permutation operators studied here is a representation of  the Walled Brauer Algebra on the space $(\mathbb{C}^d)^{\otimes n}$. Up to our knowledge it is the first result of such kind on this level of generality. We can go even further, and built a bridge between our tools,  the above-mentioned physical applications and  transposed Jucys-Murphy elements~\cite{Mu,Ju} which in their not distorted form generate commutative subalgebra of $\mathbb{C}[S(n)]$. This approach opens a new path: the opportunity for studying deformation of the permutation group $S(N)$ within a novel approach to representation theory put forward  in~\cite{Okunkov}. 

 The structure of this paper is the following. In Section~\ref{summary} we give summary of all our findings presented in the manuscript. In Section~\ref{interest} we rigorously introduce the multi-port-based teleportation schemes and discuss the quantities of interest which are entanglement fidelity and probability of success. Next, in Section~\ref{sym}  and discuss briefly the occurring symmetries. We explain the connection with the algebra of partially transposed permutation operators and the necessity of finding its irreducible components.  In Section~\ref{defs} we introduce the basic notions of the representation theory for the permutation group. We explain how to compute the basic quantities describing irreducible representations such as dimensions and multiplicities. We show how to construct an operator basis in every irreducible component. Schur-Weyl duality and notion of Young's lattice are also shortly explained. Most of the pieces of information are taken from~\cite{vershik}. In Section~\ref{preliminary} we prove a few results concerning partially transposed permutation operators. The notion of partially reduced irreducible representation (PRIR) in the generalized version concerning previous results is introduced. Using these two we prove certain summation rule for matrix elements of irreducible representations of permutations, which is up to our best knowledge not known in the literature. Finally, we present results on partial traces from the operator basis in every irreducible space of the permutation group. In Section~\ref{comm_structure} we present the main mathematical results of our paper. We construct an operator basis in every irreducible representation of the algebra of partially transposed permutation operators. Next, using this result, we compute matrix elements of a port-based teleportation operator determining the performance of teleportation schemes. We show that this object is diagonal in our basis, allowing us to determine its spectral decomposition. Having all mathematical results, in Section~\ref{detkPBT} and Section~\ref{probkPBT}, we describe deterministic and probabilistic MPBT scheme and derive expressions describing their performance. We end up by Section~\ref{diss}, where we discuss our results and present possible ways of further exploring the idea of multi-port-based teleportation schemes, for example by simultaneous optimization of the resource state and Alice's measurements.

\section{Summary of the main results}
\label{summary}
In this paper we present several results concerning twofold aspects. Firstly, we introduce tools relating the characterisation of the structure of the algebra $\mathcal{A}_n^{(k)}(d)$ to the new technical results for practical calculations in  the symmetric group $S(n)$.  Secondly, we apply our tools to characterise a class of multi-port based teleportation protocols (MPBT).
\\
{\bf Results concerning the symmetric group $S(n)$ and the algebra $\mathcal{A}_n^{(k)}(d)$:}
\begin{enumerate}[1)]
    \item In Proposition~\ref{summation0} we deliver new summation (orthogonality) rule for irreducible representations of the symmetric group $S(n)$, which is motivated by the celebrated Schur orthogonality relations~\cite{Curtis}.  This summation rule allows us for effective computations and simplifications quantities regarding MPBT protocols, especially when computing matrix elements of MPBT operator describing property of the deterministic scheme. It is also important by itself, giving deeper understanding of connection between matrix elements of a subgroup $H \subset S(n)$ and the whole group $S(n)$.
    \item We present effective tools for computing partial traces over an arbitrary number of systems from the irreducible operator basis in every irrep of $S(n)$ emerging from the Schur-Weyl duality. This is contained in Lemma~\ref{L3} and Corollary~\ref{corL3}. Up to our best knowledge these are new results on this level of generality and extending results from~\cite{Aud,Christandl_2007}.  Since these tools allow for effective calculations of partial traces in the group algebra of $S(n)$, which is often the case in quantum information science, they are of the separate interest.
    \item  We show that  the algebra $\mathcal{A}_n^{(k)}(d)$ of partially transposed permutation operators is in fact the matrix representation of the Walled Brauer Algebra on the space $(\mathbb{C}^d)^{\otimes n}$. This connection, due to~\cite{Cox1}, gives us all the ideals of the considered algebra and show how they are nested. In particular we identify the maximal ideal $\mathcal{M}$ (see Figure~\ref{structure_M}), which is the main object for further understanding multi-port based teleportation schemes. This identification is implied by the symmetries exhibit in our new teleportation protocols.
    \item We construct an orthonormal irreducible operator basis in the maximal ideal (Theorem~\ref{tmbas}). We show how the structure of the irreducible blocks looks like and explain their connection with the irreps of the symmetric groups $S(n)$ and $S(n-2k)$. In fact, this result gives us a way for constructing irreducible matrix representations of the Walled Brauer Algebra in the maximal ideal $\mathcal{M}$ on the space $(\mathbb{C}^d)^{\otimes n}$, which is the first result of such kind in the literature.
    It is analogue of the following basic result  regarding representations of $S(n)$ on $(\mathbb{C}^d)^{\otimes n}$:
    \begin{align}
        E_{ij}^\lambda=
        \frac{d_{\lambda}}{n!}\sum_{\sigma\in S(n)}
        \phi_{ji}^{\lambda}(\sigma^{-1}) V_{\sigma},
    \end{align}
    where $\lambda$ labels irreps of $S(n)$ of dimension $d_{\lambda}$, $V_{\sigma}$ is permutation operator, that permutes subsystems in 
    $(\mathbb{C}^d)^{\otimes n}$ according to permutation $\sigma\in S(n)$, and finally numbers $\phi_{ji}^{\lambda}(\sigma^{-1})$ are
    matrix elements of irreducible representation of $\sigma$.  
    The above formula is actually a general formula that works for any representation of a finite group. However, in our case, we have representation of an algebra, which is not a group algebra, and there is no such general formula.  
   
    We also construct set of projectors on irreducible blocks of the algebra $\mathcal{A}_n^{(k)}(d)$ in the maximal ideal $\mathcal{M}$, see Definition~\ref{efy} and Lemma~\ref{efy2}.
    \item 
    In the considered basis
    we find {\it matrix elements} of the basic objects for our study - namely, the permutation operators partially transposed on $k$ systems
    belonging to maximal ideal $\mathcal{M}$, as well as those permutation operators, that are not affected by partial transpose (Lemma~\ref{Vel}). 
    Our matrix elements are analogues of matrix elements 
    $\phi_{ij}^\lambda(\sigma)$ of irreps of $S(n)$ in Young-Yamanouchi basis.
    They are connected with the parameters describing irreps of the symmetric groups $S(n)$ and $S(n-2k)$.  
    This is non-trivial extension of the tools used in the Schur-Weyl duality to the case when one has to deal with symmetry of a different type (partial symmetry)- $U^{\otimes(n-k)}\otimes \overline{U}^{\otimes k}$, where the existing tools cannot be applied straightforwardly. It was possible by introducing notion of partially irreducible representations, involving concept of the induced representation and properties of subgroups. These tools allow us for effective calculations of compositions and partial traces of operators exhibiting partial symmetries, see for example Lemma~\ref{A1}, Lemma~\ref{A2} or Lemma~\ref{simple}, and surely they will find applications far beyond MPBT protocols.
\end{enumerate}
{\bf Results concerning multi-port based teleportation:}
\begin{enumerate}[1)]
    \item We investigate multi-port based teleportation schemes by identifying all their symmetries and present their connection with the algebra  $\mathcal{A}_n^{(k)}(d)$, so in fact with matrix representations of the Walled Brauer Algebra. We describe two variants, deterministic and probabilistic one. In particular, we show explicitly how operators, like signal states and measurements, encoding the performance of MPBT decompose in terms of partially transposed permutation operators (Sections~\ref{interest},~\ref{sym}). 
    \item Next, having construction of the irreducible basis in the maximal ideal $\mathcal{M}$  of the algebra $\mathcal{A}_n^{(k)}(d)$ we prove Theorem~\ref{kPBTmat} and Theorem~\ref{eig_dec_rho}. In particular these results show that the MPBT operator encoding properties of our protocols is diagonal in projectors onto irreps of the algebra $\mathcal{A}_n^{(k)}(d)$ which are known thanks to the first part of the paper. It is important to stress here that adaptation of the pre-existing tools like the Schur-Weyl duality and  the dual representation to $\overline{U}$, which led to re-computation of some known results in PBT~\cite{majenz}, when $k=1$, are not enough here. It is due to the fact that to obtain all the results one must have an orthogonal irreducible operator basis in every irreducible sector of the underlying algebra which has been not known previously.
    \item In the deterministic case we prove Theorem~\ref{Fthm} in which we present an explicit expression for entanglement fidelity $F$ of the protocol, when parties share $N=n-k$ maximally entangled states of dimension $d$ each, and use square-root measurements:
			\begin{align}
				F&=\frac{1}{d^{N+2k}}\sum_{\alpha \vdash N-k}\left(\sum_{\mu\in\alpha}m_{\mu/\alpha} \sqrt{m_{\mu}d_{\mu}}\right)^2,
\end{align}
	where $m_{\mu},d_{\mu}$ denote multiplicity and dimension of irreducible representations of $S(N)$ respectively in the Schur-Weyl duality, and $m_{\mu/\alpha}$ denotes number of paths on reduced Young's lattice in which diagram $\mu$ can be obtained from diagram $\alpha$ by adding $k$ boxes. The efficiency of the new deterministic protocol compared with deterministic PBT when teleporting a composite system is depicted in Figure~\ref{fig:test2}. In this case, we perform significantly better even than the optimal PBT. 
    \begin{figure}[h!]
	\includegraphics[width=\linewidth]{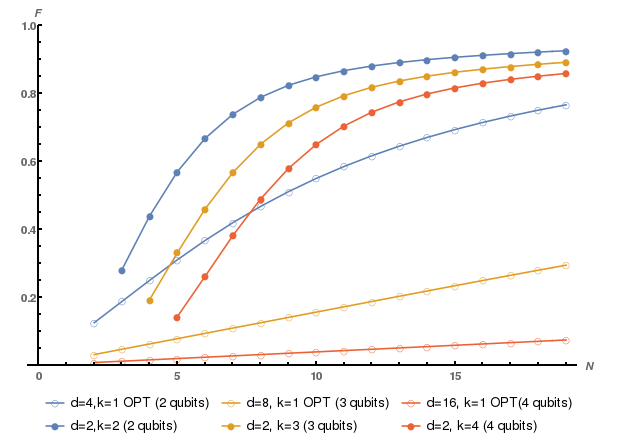}
	\caption{The performance of the deterministic version of our protocol, measured in entanglement fidelity $F$, for various choices of initial parameters which are local dimension $d$, number of ports $N$ and number of teleporting particles $k$. One can see that we achieve better performance in teleporting a state of two qubits ($d=2,k=2$) then standard PBT scheme with appropriate port dimension ($d=4,k=1$) as well as the optimal one (OPT). }
	\label{fig:test2}
\end{figure}
    \item In the probabilistic case we prove Theorem~\ref{thm_p} in which we connect probability of success $p$ with quantities describing symmetric groups $S(N)$ and $S(N-k)$:
    \be
p=\frac{k!\binom{N}{k}}{d^{2N}}\sum_{\alpha \vdash N-k}\mathop{\operatorname{min}}\limits_{\mu\in\alpha}\frac{m_{\alpha}d_{\alpha}}{\lambda_{\mu}(\alpha)},
\ee
The numbers $\lambda_{\mu}(\alpha)$ are eigenvalues of MPBT operator and $m_{\alpha}, d_{\alpha}$ denote multiplicity and dimension of the irrep labelled by $\alpha$ in the Schur-Weyl duality. The optimal measurements in this case are also derived in the same theorem. The efficiency of the new probabilistic protocol compared with probabilistic PBT when teleporting a composite system is depicted in Figure~\ref{fig:test3}. In this  case we outperform the optimal PBT scheme for $k\geq 3$.   We obtain these results by solving the dual with the primal problem and showing that they coincide, giving us the exact value of the probability and form of the optimal measurements. Exploiting symmetries of the protocol with the mathematical tools developed in this paper, we were able to solve the optimisation problem analytically, which is not the general case in the optimisation theory.
    \begin{figure}[h!]
	\includegraphics[width=\linewidth]{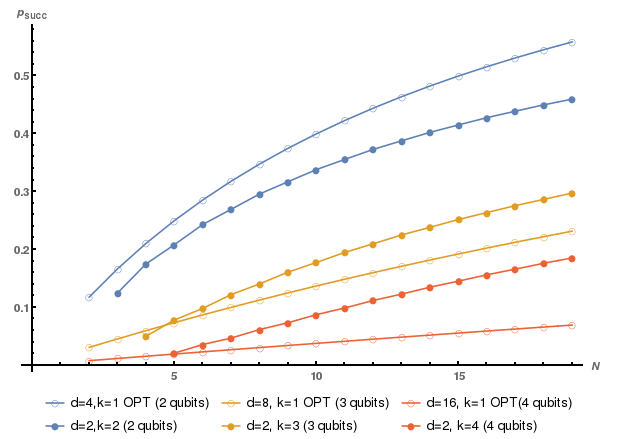}
	\caption{The performance of the probabilistic version of our protocol, measured in success probability $p$, for various choices of initial parameters which are local dimension $d$, number of ports $N$ and number of teleporting particles $k$. One can see that we start achieving better performance than the corresponding optimal PBT scheme with appropriate port dimension for a state of three qubits ($d=2,k=3$). }
	\label{fig:test3}
\end{figure}
\end{enumerate}
\section{Quantities of interest - entanglement fidelity and probability of success} 
\label{interest}
In  multi-port based teleportation protocols Alice wishes send  to Bob  an unknown composite qudit quantum state $\psi_C=\psi_{C_1C_2\cdots C_k}$, for $k\leq \floor{N/2}$, through $N$ ports, each port given as maximally entangled qudit state $|\psi^+\>=(1/\sqrt{d})\sum_{i=1}^d|ii\>$, where $d$ stands for the dimension of the underlying local Hilbert space. Both parties share so called resource state of the form $|\Psi\>=\bigotimes_{i=1}^N|\psi^+\>_{A_iB_i}$, see  Figure~\ref{FPBT}. 
Defining  
the set
\begin{multline}
\mathcal{I}:=\big\lbrace (i_1,i_2,\ldots,i_k) \ : \ \forall 1\leq l\leq k  \ i_l=1,\ldots,N  \ \\ \text{and} \  i_1\neq i_2\neq \cdots \neq i_k \big\rbrace 
\end{multline}
consisting of $k-$tuples (not necessarily ordered) denoting ports through which subsystems of the composite state $\psi_{C}$ are teleported. For example, having $N=5, k=2$ and ${\bf i}=(5,3)$ means that particle $\psi_{C_1}$ is on fifth port and $\psi_{C_2}$  on the third port.

In the next step Alice performs a joint measurement with outcomes $\mathbf{i}$ from the set $\mathcal{I}$. Every effect is described by positive operator valued measure (POVM) satisfying $\sum_{\mathbf{i}\in\mathcal{I}}\Pi_{\mathbf{i}}^{AC}=\mathbf{1}_{AC}$.
Having that we are in the position to describe a teleportation channel $\mathcal{N}$, which maps the density operators acting on $\mathcal{H}_C=\bigotimes_{i=1}^k\mathcal{H}_{C_i}$ to those acting on Bob's side:
\be
\label{ch1}
\begin{split}
	\mathcal{N}\left(\psi_{C} \right)&=\sum_{\mathbf{i}\in \mathcal{I}}\tr_{A\bar{B}_{\mathbf{i}}C}\left[ \sqrt{\Pi_{\mathbf{i}}^{AC}}\left(|\Psi\>\<\Psi|_{AB}\ot \psi_{C} \right)\sqrt{\Pi_{\mathbf{i}}^{AC}}^{\dagger}\right]\\
	                                 &=\sum_{\mathbf{i}\in \mathcal{I}}\tr_{AC}\left[\Pi_{\mathbf{i}}^{AC}\tr_{\bar{B}_{\mathbf{i}}}\bigotimes_{j=1}^N |\psi\>\<\psi|_{A_jB_j}  \ot \psi_{C} \right]_{B_{\mathbf{i}}\rightarrow \widetilde{B}}\\
	&=\sum_{\mathbf{i}\in \mathcal{I}} \tr_{AC}\left[\Pi_{\mathbf{i}}^{AC} \sigma_{\mathbf{i}}^{A\widetilde{B}}  \ot \psi_{C}\right], 
\end{split}
\ee
    where $\bar{B}_{\mathbf{i}}=\bar{B}_{i_1}\bar{B}_{i_2}\cdots \bar{B}_{i_k}$ denotes discarded  subsystems except those on positions $i_1,i_2,\ldots,i_k$ and operation $B_{\mathbf{i}}\rightarrow \widetilde{B}$ is assigning $B_{N+1}\cdots B_{N+k}$ for every index ${\bf i}$ on Bob's side, introduced for the mathematical convenience.  The states (signals) $\sigma_{\mathbf{i}}^{AB}$ or shortly $\sigma_{\mathbf{i}}$ for $\mathbf{i}\in \mathcal{I}$, from~\eqref{ch1}, are given as
		\begin{align}
\label{sigma}
\sigma_{\mathbf{i}}^{A\widetilde{B}}:= &\tr_{\bar{B}_{\mathbf{i}}}\left(P^+_{A_1B_1}\ot P^+_{A_2B_2}\ot \cdots \ot P^+_{A_NB_N} \right)_{B_{\mathbf{i}}\rightarrow \widetilde{B}} \\ =&\frac{1}{d^{N-k}}\mathbf{1}_{\bar{A}_{\mathbf{i}}}\ot P^+_{A_{\mathbf{i}}\widetilde{B}}.
\end{align}
In above $\bar{A}_{\mathbf{i}}$ has the same meaning as $\bar{B}_{\mathbf{i}}$. Then 
$P^+_{A_{\mathbf{i}}\widetilde{B}}$ is a tensor product of projectors on maximally entangled sates with respect to subsystems defined by index $\mathbf{i}$ and prescription $B_{\mathbf{i}}\rightarrow \widetilde{B}$. For example, when ${\bf i}=(5,3)$, the notation $P^+_{A_{\mathbf{i}}\widetilde{B}}$ means $P^+_{A_5B_6}\ot P^+_{A_3B_7}$. 
For the further reasons, we introduce here the following multi port-based operator given as:
\be
\label{PBT_standard}
\rho:=\sum_{\mathbf{i}\in \mathcal{I}}\sigma_{\mathbf{i}}.
\ee
In the general case in above sum we have $k!\binom{N}{k}=\frac{N!}{(N-k)!}=|\mathcal{I}|$ elements. One can see that for $k=1$ we reproduce $|\mathcal{I}|=N$ number of signals from the original PBT scheme. For $k=2$ we have $|\mathcal{I}|=N(N-1)$, for $k=3$ it is $|\mathcal{I}|=N(N-1)(N-2)$ and so on.

{\bf Deterministic version} In this version of the protocol receiver always accepts state of $k$ of $N$ ports as the teleported states. Since the ideal transmission of states is impossible, we would like to know how well we are able to preform the scheme, possibly as a function of global parameters like number of ports or local dimension. We investigate this by checking how well the teleportation channel $\mathcal{N}$ transmits quantum correlations. To do so we compute its entanglement fidelity $F$, teleporting halves of maximally entangled states
\be
\label{ent_fid}
\begin{split}
	F&=\tr\left[P^+_{\widetilde{B}C}(\mathcal{N}\ot \mathbf{1}_{D})P^+_{CD} \right]\\
	 &=\sum_{\mathbf{i}\in\mathcal{I}}\tr\left[P^+_{\widetilde{B}C} \Pi_{\mathbf{i}}^{AC}\left( \sigma_{\mathbf{i}}^{A\widetilde{B}} \ot P^+_{CD}\right)  \right]\\
	 &=\frac{1}{d^{2k}}\sum_{\mathbf{i}\in\mathcal{I}}\tr\left[\Pi_{\mathbf{i}}^{A\widetilde{B}} \sigma_{\mathbf{i}}^{A\widetilde{B}} \right],
\end{split}
\ee
where $D=D_1D_2\cdots D_k$. 
 To have explicit answer what is the value of $F$ we need to choose a specific form of POVM operators $\Pi_{\mathbf{i}}^{AC}$.
As it is explained in previous papers~\cite{ishizaka_asymptotic_2008, ishizaka_quantum_2009}, PBT scheme is equivalent to the state discrimination problem, where authors use $N$ square-root measurements for distinguishing an ensemble $\{1/N,\sigma_{i}^{AC}\}$. In our case the situation is similar and the corresponding ensemble is of the form $\{1/|\mathcal{I}|,\sigma_{\mathbf{i}}^{AC}\}$ with corresponding POVMs:
\be
\label{povm}
\forall \ \mathbf{i}\in \mathcal{I}\qquad \Pi_{\mathbf{i}}^{AC}=\frac{1}{\sqrt{\rho}}\sigma_{\mathbf{i}}^{AC}\frac{1}{\sqrt{\rho}}+\Delta,
\ee
where states $\sigma_{\mathbf{i}}^{AC}$ are given in~\eqref{sigma} and $\rho$ is the port-based operator from~\eqref{PBT1}. It can be easy seen that operator $\rho$ is not of the full rank, so inversion on the support is required. Due to this, to every component in~\eqref{povm} an additional term $\Delta$ of the form
\be
\Delta=\frac{1}{|\mathcal{S}_{n,k}|}\left(\mathbf{1}_{(\mathbb{C}^d)^{\ot n}}^{AC}-\sum_{\mathbf{i}\in\mathcal{I}}\Pi_{\mathbf{i}}^{AC} \right) 
\ee
is added. This addition ensures that all effects sum up to identity operator on whole space $(\mathbb{C}^d)^{\ot n}$. Such procedure does not change the entanglement fidelity $F$ in~\eqref{ent_fid}. Our goal is to evaluate the entanglement fidelity from~\eqref{ent_fid} with measurements given in~\eqref{povm}. The solution given in terms of group-theoretic parameters is presented in Theorem~\ref{Fthm} in Section~\ref{detkPBT}.

{\bf Probabilistic version} In this scenario transmission  sometimes fails, but whenever succeeds then fidelity of the teleported state is maximal $F=1$. The teleportation channel in the probabilistic version looks exactly the same as it is in deterministic protocol, however is non-trace preserving. This fact is due to the reason that now Alice has access to $1+|\mathcal{I}|$ POVMs, where an additional POVM $\Pi_0^{AC}$ corresponds to the failure.
To evaluate the performance of the scheme we need to calculate the average success probability of teleportation $p$, where we average over all possible input states. This leads to the following expression (see~\cite{ishizaka_quantum_2009,kopszak2021multiport} for detailed calculations):
\be
\label{p_sini}
p=\frac{1}{d^{N+k}}\sum_{\mathbf{i}\in \mathcal{I}}\tr\left(\Pi_{\mathbf{i}}^{AC} \right).
\ee 
Requirement of the unit fidelity gives strong condition on the form of the measurement applied by Alice. Namely, using argumentation presented in~\cite{ishizaka_asymptotic_2008,ishizaka_quantum_2009,Studzinski2017} all the POVMs corresponding to the success of teleportation are of the form:
\be
\label{mess}
\forall \ \mathbf{i}\in\mathcal{I}\qquad \Pi_{\mathbf{i}}^{AC}=P^+_{A_{\mathbf{i}}C}\ot \Theta_{\overline{A}_{\mathbf{i}}}.
\ee
Having expression for probability of success~\eqref{p_sini} we ask what is the maximal possible value of $p$ and what is then the optimal form the operators $\Theta_{\overline{A}_{\mathbf{i}}}$ from~\eqref{mess} ensuring mentioned maximisation. It turns out that this problem can be written as a semidefinite program (SDP).  We write down a primal problem whose solution $p^*$ lower bounds the real value of $p$, and then we write a dual problem where the solution $p_*$ is a respective upper bound for the real value of $p$.  In the boxes below we write down  explicitly primal and dual problem. 
\begin{tcolorbox}
{\bf The primal problem for probabilistic scheme} The goal is to maximise the  quantity:
\be
\label{primal}
p^*=\frac{1}{d^{N+k}}\sum_{\mathbf{i}\in \mathcal{I}}\tr \Theta_{\overline{A}_{\mathbf{i}}},
\ee
with respect to the following constraints:
\be
\label{con1}
(1) \ \Theta_{\overline{A}_{\mathbf{i}}}\geq 0,\qquad (2) \ \sum_{\mathbf{i}\in \mathcal{I}}P^+_{A_{\mathbf{i}}C}\ot \Theta_{\overline{A}_{\mathbf{i}}}\leq \mathbf{1}_{AC}.
\ee 
\end{tcolorbox}
\begin{tcolorbox}
{\bf The dual problem for probabilistic scheme}  The dual problem is to minimise the quantity
\be
\label{dual0}
p_*=\frac{1}{d^{N+k}}\tr \Omega
\ee
with respect to the following constraints
\be
\label{dual}
(1) \ \Omega \geq 0, \qquad (2) \ \forall \mathbf{i}\in\mathcal{I} \ \tr_{A_{\mathbf{i}}C}\left(P^+_{A_{\mathbf{i}}C}\Omega\right)\geq \mathbf{1}_{N-k}, 
\ee
where the operator $\Omega$ acts on $N+k$ systems, the identity $\mathbf{1}_{N-k}$ acts on all systems but $A_{\mathbf{i}}$ and $C$.
\end{tcolorbox}
The solutions of the above primal and the dual problem in terms of group-theoretic parameters are presented in Theorem~\ref{thm_p} in Section~\ref{probkPBT}. For potential reader who would like to learn more about the concept of SDP, formulation of the respective primal and dual problem, we refer  to book~\cite{Boyd}.

\section{Symmetries in multi port-based teleportation} 
\label{sym}
In every variant of (multi) port-based teleportation protocols we distinguish two type of symmetries. One is connected with covariance and invariance with respect to the symmetric group $S(n-k)$, while the second one with invariance with respect to the action of $U^{\otimes (n-k)}\otimes \overline{U}^{\otimes k}$. We now shall describe briefly connection of these two types of symmetries with the operators describing analysed teleportation schemes.

Let us take index $\mathbf{i}_0$ such that $\mathbf{i}_0=(N-2k+1,N-2k+2,\ldots,N-k)$, then having $n=N+k$ signal $\sigma_{\mathbf{i}_0}$ can be written as
\be
\begin{split}
	\sigma_{\mathbf{i}_0}&=\frac{1}{d^{n-2k}}\mathbf{1}_{\bar{A}_{\mathbf{i}_0}}\ot P^+_{n-2k+1,n}\ot P^+_{n-2k+2,n-1}\ot \cdots \ot P^+_{n-k,n-k+1} \\
	&=\frac{1}{d^{n-k}}\mathbf{1}_{\bar{A}_{\mathbf{i}_0}}\ot V^{t_n}_{n-2k+1,n}\ot V^{t_{n-1}}_{n-2k+2,n-1}\ot \cdots \ot V^{t_{n-k+1}}_{n-k,n-k+1},
\end{split}
\ee
where by $t_n,t_{n-1}$ etc. we denote the partial transpositions with respect to particular subsystem and by $V_{r,s}$ the permutation operator between system $r$ and $s$ (since now we drop off indices for $A,B$ unless they necessary), for the $V_{r,s}^{t_s}$ operator is proportional to $P^+_{r,s}$. Further we assume whenever it is necessary that permutation operators  are properly embedded in whole $(\mathbb{C}^d)^{\ot n}$ space so we will write just $V_{r,s}$ instead of $V_{r,s}\ot \mathbf{1}_{\bar{r},\bar{s}}$. Moreover for the signal $\sigma_{\mathbf{i}_0}$ we introduce simpler notation
\be
\label{signal2}
\begin{split}
	\sigma_{\mathbf{i}_0}&=\frac{1}{d^{n-k}}\mathbf{1}\ot V^{t_n}_{n-2k+1,n}\ot V^{t_{n-1}}_{n-2k+2,n-1}\ot \cdots \ot V^{t_{n-k+1}}_{n-k,n-k+1}\\
	&\equiv \frac{1}{d^{N}} V^{(k)},
\end{split}
\ee
where 
\be
\begin{split}
	&V^{(k)}\equiv \mathbf{1}\ot V^{t_n}_{n-2k+1,n}\ot V^{t_{n-1}}_{n-2k+2,n-1}\ot \cdots \ot V^{t_{n-k+1}}_{n-k,n-k+1},\\
	&(k)\equiv t_n \circ t_{n-1}\circ \cdots \circ t_{n-k+1},
\end{split}
\ee
$\circ$ denotes composition of maps, and $\mathbf{1}$ denotes identity acting on the space untouched by tensor product of projectors on maximally entangled states. Form the definition of the signals $\sigma_{\mathbf{i}}$ in~\eqref{sigma} and form of $\sigma_{\mathbf{i}_0}$ from~\eqref{signal2} we can deduce that PBT operator $\rho$ can be written as
\be
\label{PBT1}
\rho=\sum_{\mathbf{i}\in \mathcal{I}}\sigma_{\mathbf{i}}=\frac{1}{d^N}\sum_{\tau \in \mathcal{S}_{n,k}}V_{\tau^{-1}}V^{(k)}V_{\tau},
\ee
where sum runs over all permutations $\tau$ from the coset $\mathcal{S}_{n,k}:= \frac{S(n-k)}{S(n-2k)}$, and $V_{\tau}$ the permutation operator corresponding to the permutation $\tau$. Due to the construction we have that $|\mathcal{I}|=\left|\frac{S(n-k)}{S(n-2k)}\right|$. Moreover, it is easy to notice that any signal state satisfies:
\begin{equation}
\label{rel1}
\forall \ \tau\in \mathcal{S}_{n,k} \qquad V_{\tau}\sigma_{\mathbf{i}}V_{\tau^{-1}}=\sigma_{\tau(\mathbf{i})}.
\end{equation}
Let us notice that the operator $\rho$ is invariant with respect to action of any permutation from $S(n-k)$ acting on the first $n-k$ systems:
\be
\label{rel2}
\forall \ \tau\in S(n-k) \qquad V_{\tau}\rho V_{\tau^{-1}}=\rho.
\ee
In particular, relation~\eqref{rel1} and~\eqref{rel2} imply covariance of the SRM measurements given in~\eqref{povm} with respect to the coset $\mathcal{S}_{n,k}$. The same type of covariance we require for POVMs in the probabilistic scheme from~\eqref{mess}. 

 We have also the second kind of symmetries. Notice that all operators $\sigma_{\mathbf{i}}$ as well PBT operator from~\eqref{PBT_standard} are invariant with respect to action of $U^{\ot (n-k)}\ot \overline{U}^{\ot k}$, where the bar denotes the complex wise conjugation and $U$ is an element of unitary group $\mathcal{U}(d)$. This observation follows from the structure of the signal states and the fact that every bipartite maximally entangled state $P^+_{ij}$, between system $i$ and $j$, is $U\otimes \overline{U}$ invariant. 
 
 This property with expression~\eqref{signal2} means that basic elements describing the performance of the  presented teleportation protocol belong to the algebra $\mathcal{A}^{(k)}_n(d)$ of partially transposed permutation operators with respect to $k$ last subsystems.
For $k=1$ we reduce to the known case, and standard $d$ dimensional port-based teleportation introduced in~\cite{Studzinski2017} and the algebra $\mathcal{A}^{(1)}_n(d)$ discussed in~\cite{Stu1,Moz1}. In this particular case we have  $\mathcal{S}_{n,1}=S(n-1)/S(n-2)$ and all permutations $\tau$ are of the form of transpositions $(a,n)$ for $a=1,\ldots,n-1$, and $|S(n-1)/S(n-2)|=n-1$.

The operation of partial transposition changes significantly properties of the operators under consideration, making the resulting set of operators no longer the group algebra of the symmetric group $S(n)$. To see it explicitly, let us consider a swap operator $V_{i,j}$ interchanging systems on positions $i$ and $j$. It is obvious that by applying the swap operator twice, we end up with an identity operator. However, applying the partial transposition to $V_{i,j}$, the swap operator is mapped to the operator $dP^+_{ij}$, which is proportional to maximally entangled state between respective systems. Applying partially transposed swap operator twice, we end up with $d^2P^+_{ij}$, since $P^+_{ij}P^+_{ij}=P^+_{ij}$. This property makes our further analysis more complex, and direct application of the standard methods from the representation theory of the group algebra of the symmetric group $S(n)$ is insufficient here.

In next sections we introduce notations and definitions and construct irreducible orthonormal basis of the algebra $\mathcal{A}^{(k)}_n(d)$ and formulate auxiliary lemmas required to spectral analysis of the operator $\rho$ and describing the performance of the protocol. 
\section{Notations and Definitions}
\label{defs}
For a given natural number $n$ we can define its partition $\mu$ in the following way
\be
\mu=(\mu^{(1)},\mu^{(2)},\ldots,\mu^{(l)})
\ee

\text{such that}
\begin{multline}
\forall_{1\leq i\leq l} \  \mu^{(i)}\in \mathbb{N}, \quad  \mu^{(1)}\geq \mu^{(2)} \geq \cdots \geq \mu^{(l)}\geq 0,\\\sum_{i=1}^l\mu^{(i)}=n.
\end{multline}
The Young frame associated with partition $\mu$ is the array formed by $n$ boxes with $l$ left-justified rows. The $i$-th row contains exactly $\mu^{(i)}$ boxes for all $i=1,2,\ldots,l$. 
 Further, we denote Young diagrams by the Greek letters. The set of all Young diagrams, with up to $n\in\mathbb{N}$ boxes, is denoted as $\mathbb{Y}_n$. The restriction to the set of Young diagrams with no more then $d$ rows is denoted as $\mathbb{Y}_{n,d}$. We endow $\mathbb{Y}_n,\mathbb{Y}_{n,d}$ with a structure of a partially ordered set by setting, for $\mu=(\mu^{(1)},\mu^{(2)},\ldots,\mu^{(l)})\vdash n$ and $\alpha=(\alpha^{(1)},\alpha^{(2)},\ldots,\alpha^{(s)})\vdash n-k$, 
\be
\alpha \preceq \mu, 
\ee
if $\mu^{(i)}\geq \alpha^{(i)}$ for all $i=1,2,\ldots,l$. If $\alpha \preceq \mu$ we denote by $\mu/\alpha$ the array, called also a skew shape, obtained by removing from the Young frame $\mu$ the boxes of the Young frame of $\alpha$. We have illustrated this procedure by an example presented in Figure~\ref{skewshape}. 
\begin{figure}[h]
	\includegraphics[width=\linewidth]{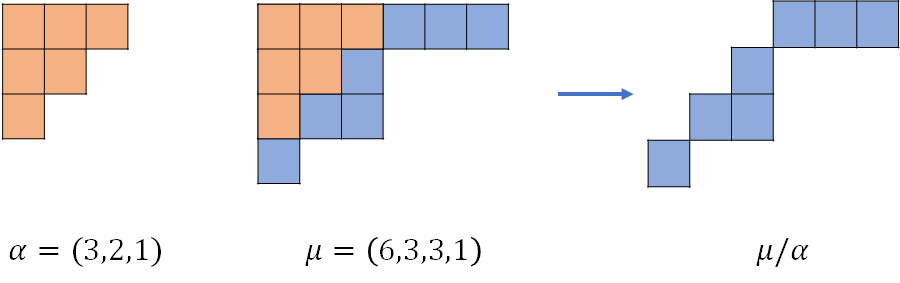}
	\caption{Graphics presents construction of a skew shape $\mu/\alpha$ for Young frames $\mu=(6,3,3,1)$ and $\alpha=(3,2,1)$. }
	\label{skewshape}
\end{figure}
For any $\alpha, \mu \in \mathbb{Y}_n$ we say that $\mu$ covers $\alpha$, or $\alpha$ is covered by $\mu$ if $\alpha \preceq \mu$ and
\be
\alpha \preceq \nu \preceq \mu, \quad \nu \in \mathbb{Y}_n \Rightarrow \nu=\alpha \ \text{or} \ \nu=\mu.
\ee
In other words, $\mu$ covers $\alpha$ if and only if $\alpha \preceq \mu$ and $\mu/\alpha$ consists of at least a single box.  Later we use an equivalent symbol  $\mu \in \alpha$ to denote Young diagrams $\mu \vdash n$ obtained from Young diagrams $\alpha \vdash n-k$ by adding $k$ boxes. While by the symbol $\alpha \in \mu$ we denote Young diagrams $\alpha \vdash n-k$ obtained from Young diagrams $\mu \vdash n$ by subtracting $k$ boxes. Informally it means that a Young diagram with $n-k$ boxes is contained in a Young diagram with $n$ boxes.
 Having the concept of Young diagram and sets $\mathbb{Y}_n, \mathbb{Y}_{n,d}$ we define Young's lattice and its reduced version (see Figure~\ref{bratteli}). 
\begin{figure}[h]
	\includegraphics[width=\linewidth]{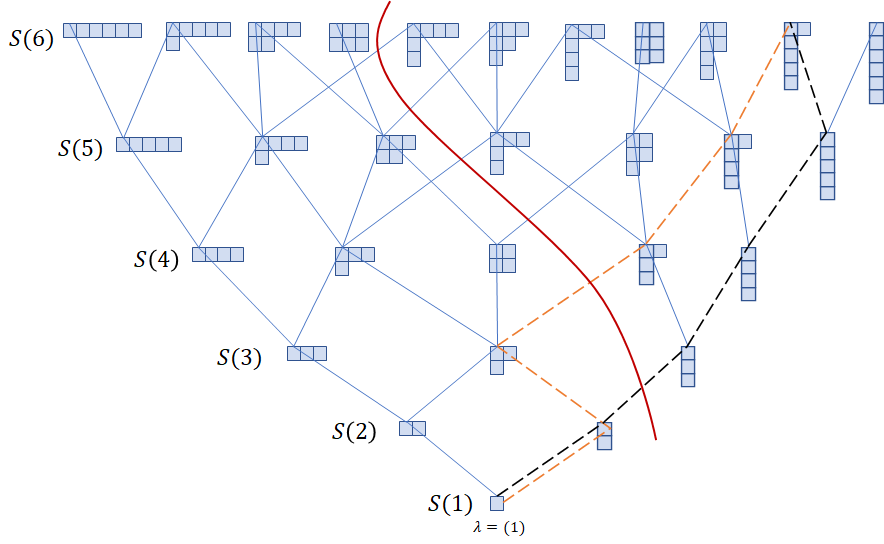}
	\caption{The Young's lattice $\mathbb{Y}_6$, i.e. with six consecutive layers labelled by permutation groups from $S(1)$ to $S(6)$. By orange and black dashed lines we depict two possible paths from irrep $\lambda=(1)$ of $S(1)$ to irrep $\lambda'=(2,1,1,1,1)$ of $S(6)$. The reduced Young's lattice $Y_{6,2}$, i.e. for $d=2$ is defined by all diagram on the right-hand side of the red line.}
	\label{bratteli}
\end{figure}
The Young's lattice arises when we construct subsequent Young diagrams by adding boxes one by  one. In this way we obtain subsequent layers of Young diagrams for growing $n$. 
We connect a diagram  with a subsequent diagram by an edge, that is obtained by 
adding a box.  More formally the Young's lattice of $\mathbb{Y}_n$ is the non-oriented graph with vertex set $\mathbb{Y}_n$ and an edge from $\lambda$ to $\mu$ if and only if $\lambda$ covers $\mu$. The same definition applies for Young's lattice of $\mathbb{Y}_{n,d}$, but we remove all Young diagrams with more than $d$ rows. A path $r_{\mu/\alpha}$ in the Young's lattice is a sequence $r_{\mu/\alpha}=(\mu\equiv\mu_n \vdash n \rightarrow \mu_{n-1} \vdash n-1 \rightarrow \cdots \rightarrow \alpha \equiv \mu_{n-k} \vdash n-k)$, for some $k\in\mathbb{N}$ and $k<n$.
The integer number $m_{\mu/\alpha}$ is the total lengths of all paths from $\mu$ to $\alpha$.

All irreducible representations (irreps) of $S(n)$ are labelled by Young diagrams with $n$ boxes denoted as $\alpha \vdash_d n$, where $d$ is a natural parameter, meaning we take into account diagrams with at most $d$ rows.
The dimension $d_{\alpha}$ of the irrep $\alpha$ is given by the hook length formula
\be
d_{\mu}=\frac{n!}{\mathop{\prod}\limits_{(i,j)\in \mu}h_{\mu}(i,j)},
\ee
where $h_{\mu}(i,j)$ is so called the hook length of the hook with corner at the box $(i,j)$ given as one plus the number of boxes below $(i,j)$ plus number of boxes to the right of $(i,j)$.
Multiplicity $m_{\mu, d}$ of every irrep $\alpha \vdash_d n$ is characterised by Weyl dimension formula, saying that
\be
m_{\mu, d}=\mathop{\prod}\limits_{1\leq i\leq j\leq d}\frac{\mu_i-\mu_j+j-i}{j-i}.
\ee
Later on we suppress notation to $\alpha \vdash n$ and $m_{\mu}$, having in mind the dependence of the natural parameter $d$, playing later the role of local dimension of the space $\mathbb{C}^d$.

Having commuting representations of $S(n)$ and $\mathcal{U}(d)$ on $(\mathbb{C}^d)^{\ot n}$, acting by permuting the tensor factors, and multiplication by $U^{\ot n}$ respectively, we can decompose the space $(\mathbb{C}^d)^{\ot n}$, using Schur-Weyl duality~\cite{Wallach} into direct sum of irreducible subspaces as follows:
\be
\label{SW}
(\mathbb{C}^d)^{\ot n}\cong \bigoplus_{\alpha \vdash_d n} \mathcal{U}_{\alpha}\ot \mathcal{S}_{\alpha}.
\ee
In the above $\mathcal{S}_{\alpha}$ are representation spaces for the permutation groups $S(n)$, while $\mathcal{U}_{\alpha}$ are representation spaces of $\mathcal{U}(d)$.
In Schur basis producing the decomposition~\eqref{SW} we can define in every space $\mathcal{S}_{\alpha}$ an orthonormal operator basis $E_{ij}^{\alpha}$, for $i,j=1,\ldots,d_{\alpha}$, separating the multiplicity and representation space of permutations respectively. Namely we have
\be
\label{schur_Eij}
E_{ij}^{\alpha}=\mathbf{1}_{\alpha}^{\mathcal{U}}\otimes |\alpha,i\>\<\alpha,j|\equiv \mathbf{1}_{\alpha}^{\mathcal{U}}\otimes |i\>\<j|_{\alpha}.
\ee 
We can also use representation of $E^{\alpha}_{ij}$ on the space $(\mathbb{C}^d)^{\ot n}$, which is of the form
\be
\label{Eij}
E_{ij}^{\alpha}=\frac{d_{\alpha}}{n!}\sum_{\tau \in S(n)}\phi_{ji}^{\alpha}(\tau^{-1})V_{\tau},
\ee
where $\phi_{ji}^{\alpha}(\tau^{-1})$ denotes matrix element irreducible representation of the permutation $\tau^{-1}\in S(n)$.
The operators from~\eqref{schur_Eij} have the following properties
\be
\label{tr_prop}
E_{ij}^{\alpha}E_{kl}^{\beta}=\delta^{\alpha \beta}\delta_{jk}E^{\alpha}_{il},\quad \tr E^{\alpha}_{ij}=m_{\alpha}\delta_{ij}.
\ee
Let us observe that operators $E^{\alpha}_{ii}$ are projectors. Action of the operators $E_{ij}^{\alpha}$ on an arbitrary permutation operator $V_{\sigma}$, from the left and from the right-hand side, for $\sigma\in S(n)$ is given by
\be
\label{actionE}
E_{ij}^{\alpha}V_{\sigma}=\sum_k \varphi_{jk}^{\alpha}(\sigma)E_{ik}^{\alpha},\quad V_{\sigma}E_{ij}^{\alpha}=\sum_k\varphi_{ki}^{\alpha}(\sigma)E_{kj}^{\alpha}.
\ee
Using this basis we can write matrix representation of a given permutation $\tau^{-1}\in S(n)$, on every irreducible space labelled by $\alpha \vdash n$ as
\be
\label{blaa}
\phi^{\alpha}(\tau^{-1})=\sum_{ij}\phi_{ij}^{\alpha}(\tau^{-1})E_{ij}^{\alpha}.
\ee
 Moreover, using this operators we construct Young projectors, the projectors on components $\mathcal{U}_{\alpha}\ot \mathcal{S}_{\alpha}$ from~\eqref{SW}:
\be
\label{def_P}
P_{\alpha}=\sum_i E^{\alpha}_{ii}=\frac{d_{\alpha}}{n!}\sum_{\tau \in S(n)}\chi^{\alpha}(\tau^{-1})V_{\tau}.
\ee
The numbers $\chi^{\alpha}(\tau^{-1})=\sum_i \phi_{ii}^{\alpha}(\tau^{-1})$ are irreducible characters.

Sometimes instead of $|\alpha,i\>$ we write $|\alpha,i_{\alpha}\>$ or just $|i_{\alpha}\>$. Defining $c=\sum_{\alpha}m_{\alpha}d_{\alpha}$ any operator  $X\in M(c\times c, \mathbb{C})$ can be written using elements $\{E_{ij}^{\alpha}\}$ as $X=\sum_{\alpha}\sum_{i,j=1}^{d_{\alpha}}x_{ij}^{\alpha}E_{ij}^{\alpha}.$
Considering $n-$particle system, by writing $V_{n-1,n}$ we understand $\mathbf{1}_{1\ldots n-2}\otimes V_{n-1,n}$, and similarly for other operators. The operator $\mathbf{1}_{1\ldots n-2}$ is identity operator on first $n-2$ particles.
\section{Preliminary Mathematical Results}
\label{preliminary}
\subsection{Partial trace over Young projectors}
In this section we present a set of auxiliary lemmas which are crucial for the presentation in the further sections. 
Introducing notation $\tr_{(k)}=\tr_{n-2k+1,\ldots,n-k}$, denoting the partial trace over a set $\{n-2k+1,\ldots,n-k\}$ of particles,  and  recalling the notation
\be
\label{parV}
V^{(k)}:=V^{t_{n}}_{n-2k+1,n}V^{t_{n-1}}_{n-2k+2,n-1}\cdots V^{t_{n-k+1}}_{n-k,n-k+1}
\ee
we start from formulation of the following:
\begin{fact}
	\label{kpartr}
	For any operator acting on $n-k$ systems we have the following equality
	\be
	V^{(k)}X\ot \mathbf{1}_{n-k+1\ldots n}V^{(k)}=\tr_{(k)}(X)V^{(k)},
	\ee
	where $\mathbf{1}_{n-k+1\ldots n}$ is the identity operator on $k$ last subsystems, while the operator $X$ on first $n-k$.
\end{fact}
In particular cases, when $k=1,2$, we have respectively
\be
\label{kpareq}
\begin{split}
	V^{(1)}X\otimes \mathbf{1}_nV^{(1)}&=\tr_{n-1}(X)V^{(1)},\\ V^{(2)}X\ot \mathbf{1}_{n-1,n}V^{(2)}&=\tr_{n-3,n-2}(X)V^{(2)}.
\end{split}
\ee

Now we prove Fact~\ref{kpartr}:
\begin{proof}
	It is enough to show that expression~\eqref{kpareq} holds for $k=1$ and then use the argumentation below iteratively. Using identity $V^{(1)}=dP_+$, where $P_+$ is the projector on maximally entangled state and $d$ is the dimension of local Hilbert space, we write
	\be
	\begin{split}
		V^{(1)}&(X\ot \mathbf{1}_n)V^{(1)}=\\
					 &d^2(\mathbf{1}_{1,\dots,n-2}\ot P_+)\left(\sum_{ij}X_{ij}^{1,\dots,n-2}\ot e_{ij}^{(n-1)}\ot \mathbf{1}_n \right)(\mathbf{1}_{1,\dots,n-2}\ot P_+), 
	\end{split}
	\ee 
	where  $\{e_{ij}^{(n-1)}\}$ is standard operator  basis on subsystem $n-1$ and $X_{ij}^{1,\dots,n-2}$ are operators on subsystems from $1$ to $n-2$.	Using explicit form of $P_+=(1/d)\sum_{k,l}e_{kl}^{(n-1)}\ot e_{kl}^{(n)}$ and $X_{ij}^{(1,\dots,n-2)}$ we get~\eqref{eq:long}, given at the top of the next page,
	where $x_{i_1j_1,\dots,i_{n-1}j_{n-1}}$ denotes matrix elements of X.
\newcounter{MYtempeqncnt}

\begin{figure*}[t!]
	\normalsize
	\setcounter{MYtempeqncnt}{\value{equation}}
	\setcounter{equation}{40}
	\begin{align}
		V^{(1)}(X\ot \mathbf{1}_n)V^{(1)}&=\sum_{kl}\left(\mathbf{1}_{1,\dots,n-2}\otimes e_{kl}^{(n-1)} \otimes e_{kl}^{(n)}\right) \left(\sum_{ij}X_{ij}^{1,\dots,n-2}\ot e_{ij}^{(n-1)}\ot \mathbf{1}_n \right)\sum_{pq} \left(\mathbf{1}_{1,\dots,n-2}\otimes e_{pq}^{(n-1)} \otimes e_{pq}^{(n)}\right) \label{eq:long}\\
	&=\sum_{kq}  \left(\sum_{\substack{i_1,\dots,i_{n-2},p,\\ j_1, \dots, j_{n-2}, p}} x_{i_1j_1,\dots,i_{n-2}j_{n-2},pp} e_{i_1j_1}^{(1)}\otimes\dots\otimes e_{i_{n-2}j_{n-2}}^{(n-2)}\right)\otimes e_{kq}^{(n-1)}\otimes e_{kq}^{(n)}\nonumber\\
	&=\sum_{kq}\left(\sum_p X_{pp}^{1,\dots,n-2}\right)\otimes e_{kq}^{(n-1)}\otimes e_{kq}^{(n)}\nonumber\\
	&=\tr_{n-1}(X)V^{(1)}\nonumber
	\end{align}
	\setcounter{equation}{\value{MYtempeqncnt}+1}
	\hrulefill
	\vspace*{4pt}
\end{figure*}
\end{proof}
\begin{fact}
	\label{f1a}
	Let $\{|i\>_{\alpha}\}_{i=1}^{d_{\alpha}}$ denote a basis in an irrep $\alpha \vdash n$ of dimension $d_{\alpha}$. Then for any operator  $X^{\alpha}=\sum_{ij}x^{\alpha}_{ij}|i\>\<j|_{\alpha}$ acting on the space $\mathbb{C}^{d_{\alpha}}$, we have
	\be
	\label{ef1}
	\sum_{\tau\in S(n)}\tr \left(X^{\alpha} \phi^{\alpha}(\tau^{-1}) \right)V_{\tau}=\frac{n!}{d_{\alpha}}\mathbf{1}_{\alpha}^U\otimes X^{\alpha}.
	\ee
\end{fact}

\begin{proof}
	Inserting expression~\eqref{blaa} into~\eqref{ef1} we obtain
	\be
	\begin{split}
		\sum_{\tau\in S(n)}\tr \left(X^{\alpha} \phi^{\alpha}(\tau^{-1}) \right)V_{\tau}& 
		=\sum_{\tau\in S(n)}\tr \left(X^{\alpha}\sum_{ij}\phi_{ij}^{\alpha}(\tau^{-1})E_{ij}^{\alpha}\right)V_{\tau}\\
		& =\sum_{ij}\tr\left(X^{\alpha} E_{ij}^{\alpha}\right)\sum_{\tau\in S(n)}\phi_{ij}^{\alpha}(\tau^{-1})V_{\tau}\\
		& =\frac{n!}{d_{\alpha}}\sum_{ij}\tr\left(X^{\alpha} E_{ij}^{\alpha}\right)E_{ji}^{\alpha}.
	\end{split}
	\ee
	In Schur basis every operator $E_{ij}^{\alpha}$ has a form $\mathbf{1}^U_{\alpha}\otimes |i\>\<j|_{\alpha}$. This allows us to write
	\be
	\begin{split}
		&\frac{n!}{d_{\alpha}}\sum_{ij}\tr\left( X^{\alpha} E_{ij}^{\alpha}\right)E_{ji}^{\alpha}\\
		&=\frac{n!}{m_{\alpha}d_{\alpha}}\sum_{ij}\tr\left[ \left( \mathbf{1}^{U}_{\alpha}\otimes X^{\alpha}\right) \left(  \mathbf{1}^{U}_{\alpha}\otimes |i\>\<j|_{\alpha}\right) \right] \left(\mathbf{1}^{U}_{\alpha}\otimes |j\>\<i|_{\alpha} \right)\\
		&=\frac{n!}{d_{\alpha}}\mathbf{1}^U_{\alpha}\ot \tr\left(X^{\alpha}|i\>\<j|_{\alpha}\right) |j\>\<i|_{\alpha}=\frac{n!}{d_{\alpha}}\mathbf{1}^U_{\alpha}\ot X^{\alpha}.
	\end{split}
	\ee
\end{proof}

Now, let us introduce the following objects:
\be
\label{obj}
\widetilde{\phi}^{\mu}(a,n):= d^{\delta_{a,n}}\phi^{\mu}(a,n),\quad \text{and} \quad \widetilde{V}_{a,n}:= d^{\delta_{a,n}}V_{a,n},
\ee
where $\phi^{\mu}(a,n)$ is a matrix representation of permutation $V_{a,n}$ on irrep $\mu \vdash n$.
 Having that we can formulate the following
\begin{lemma}
	\label{object}
	Let us denote by $\mathbf{1}^{\mu}_{\gamma}$ the identity on an irrep $\gamma$ of $S(n-1)$ contained in irrep $\mu$ of $S(n)$, then we have the following restriction of $\widetilde{\phi}^{\mu}(a,n)$ to irrep $\beta$ of $S(n-1)$
	\be
	\left[ \sum_{a=1}^{n}\widetilde{\phi}^{\mu}(a,n)\right]_{\beta}=x^{\mu}_{\beta}\mathbf{1}^{\mu}_{\beta}\quad with \quad x_{\beta}^{\mu}=n\frac{m_{\mu}d_{\beta}}{m_{\beta}d_{\mu}}.
	\ee
\end{lemma}

\begin{proof}
	Consider
	\be
	\sum_{a=1}^{n}\widetilde{\phi}^{\mu}(a,n)=\sum_{a=1}^{n-1}\phi^{\mu}(a,n)+d\mathbf{1}_{\mu}
	\ee
	which is clearly invariant with respect to $S(n-1)$. Hence it admits the decomposition
	\be
	\sum_{a=1}^{n}\widetilde{\phi}(a,n)=\sum_{\gamma \in \mu}x_{\gamma}^{\mu}\mathbf{1}^{\mu}_{\gamma}
	\ee
	for some $x_{\gamma}^{\mu}\in \mathbb{C}$. The restriction for  chosen irrep $\beta \in \mu$ reduces the above to
	\be
	\label{x3}
	\left[ \sum_{a=1}^{n}\widetilde{\phi}^{\mu}(a,n)\right]_{\beta}=x^{\mu}_{\beta}\mathbf{1}^{\mu}_{\beta}.
	\ee
	Now our goal is to compute the unknown coefficients $x_{\gamma}^{\mu}$.  To do so let us first observe that we can write every projector $P_{\mu}$ in terms of coset elements $\phi(a, n)$ and permutations from $S(n-1)$. Indeed we have
	\be
	\label{P1}
	\begin{split}
		P_{\mu}&=\frac{d_{\mu}}{n!}\sum_{\sigma \in S(n)}\tr\left[  \phi^{\mu}(\sigma^{-1})\right] V_{\sigma}\\
					 &=\frac{d_{\mu}}{n!}\sum_{a=1}^{n}\sum_{\tau \in S(n-1)}\tr\left[\phi^{\mu}((a, n)\circ \tau^{-1}) \right]V_{a,n}V_{\tau}.
	\end{split}
	\ee
	Since every representation is a homomorphism we have $\phi^{\mu}((a, n)\circ \tau^{-1}) =\phi^{\mu}((a, n))\phi^{\mu}(\tau^{-1})$. Moreover, because $\tau \in S(n-1)$ and $\mu \vdash n$, representation $\phi^{\mu}(\tau^{-1})$ has to be block diagonal in $\alpha \vdash n-1$:
	\be
	\label{phi}
	\phi^{\mu}(\tau^{-1})=\bigoplus_{\alpha \in \mu}\phi^{\alpha}(\tau^{-1}),
	\ee
	where the symbol $\alpha \in \mu$ denotes all Young frames obtained from $\mu$ by removing a single box. Denoting by $\mathbf{1}_{\mu},\mathbf{1}_{\alpha}$ identities on irreps $\mu \vdash n$ and $\alpha \vdash n-1$ respectively, for which $\alpha \in \mu$ holds, we write $\mathbf{1}_{\mu}=\bigoplus_{\alpha \in \mu}\mathbf{1}_{\alpha}$. Applying this identity together with~\eqref{phi} to equation~\eqref{P1} we rewrite as
	\be
	\label{P2}
	\begin{split}
		P_{\mu}&=\frac{d_{\mu}}{n!}\sum_{a=1}^{n}V_{a,n}\sum_{\alpha \in \mu}\sum_{\tau \in S(n-1)}\tr\left(\left[\phi^{\mu}(a,n) \right]_{\alpha}\phi^{\alpha}(\tau^{-1}) \right) V_{\tau},
	\end{split}
	\ee
	where $\left[\phi^{\mu}(a,n) \right]_{\alpha}\equiv \mathbf{1}_{\alpha}\phi^{\mu}(a,n)\mathbf{1}_{\alpha}$. Using Fact~\ref{f1a} to expression~\eqref{P2} we have
	\be
	\label{Pan}
	\begin{split}
		P_{\mu}&=\frac{1}{n}\sum_{a=1}^{n}V_{a,n}\sum_{\alpha \in \mu}\frac{d_{\mu}}{d_{\alpha}}\left(\mathbf{1}_{\alpha}^U\ot \left[\phi^{\mu}(a,n) \right]_{\alpha} \right). 	
	\end{split}
	\ee
	Having~\eqref{Pan} and definitions~\eqref{obj}, together with~\eqref{x3}, and $\tr_n V(a,n)=d^{\delta_{a,n}}\mathbf{1}_{1\ldots n-1}$, we write
	\be
	\label{x2}
	\begin{split}
		\tr_{n-1}P_{\mu}& =\frac{1}{n}\sum_{\beta \in \mu=}\frac{d_{\mu}}{d_{\beta}}\mathbf{1}_{\beta}^{U}\ot \left[\sum_{a=1}^{n}\widetilde{\phi}^{\mu}(a,n) \right]_{\beta}\\
		& =\frac{1}{n}\sum_{\beta \in \mu}\frac{d_{\mu}}{d_{\beta}}x_{\beta}^{\mu}\mathbf{1}_{\beta}^U\ot \mathbf{1}_{\beta}^S=\frac{1}{n}\sum_{\beta \in \mu}\frac{d_{\mu}}{d_{\beta}}x_{\beta}^{\mu}P_{\beta}.
\end{split}
	\ee
	Using~\eqref{x2}, property $\tr(P_{\beta}P_{\mu})=m_{\mu}d_{\beta}$, for $\mu \vdash n, \beta \vdash n-1$, and
	\be
	\begin{split}
		\tr\left( P_{\beta}P_{\mu}\right)&=\frac{1}{n}\sum_{\gamma \in \mu}\frac{d_{\mu}}{d_{\gamma}}x_{\gamma}^{\mu}\tr\left(P_{\gamma}P_{\beta} \right)\\
																		 &=\frac{1}{n}\frac{d_{\mu}}{d_{\beta}}x^{\mu}_{\beta}d_{\beta}m_{\beta}=\frac{1}{n}d_{\mu}m_{\beta}x^{\mu}_{\beta}  
\end{split}
	\ee
	we deduce that
	\be
	\label{coeff}
	x_{\beta}^{\mu}=n\frac{m_{\mu}d_{\beta}}{m_{\beta}d_{\mu}}.
	\ee
	This finishes the proof.
\end{proof}

\subsection{A new summation rule for irreducible representations and PRIR (Partially Reduced Irreducible Representation) notation}
\label{prir}
Let $H\subset S(n)$ be an arbitrary subgroup of $S(n)$ with transversal $%
T=\{\tau _{k}:k=1,\ldots,\frac{n!}{|H|}\}$, i.e. we have 
\be
S(n)=\bigcup _{k=1}^{\frac{n!}{|H|}}\tau _{k}H. 
\ee
For the further purposes, we can also introduce simplified notation
\begin{notation}
\label{not0}
Let us take $\mu \vdash n$ and $\alpha \vdash n-k$, for $k<n$.
By index $r_{\mu/\alpha}$ we denote a path on Young's lattice from diagram $\mu$ to $\alpha$. This path is uniquely determined by choosing a chain of covered young frames from $\mu$ to $\alpha$, differencing by one box in each step:
\be
r_{\mu/\alpha}=(\mu,\mu_{n-1},\ldots,\mu_{n-k+1},\alpha)
\ee 
and
\be
\mu \ni \mu_{n-1} \ni \cdots \ni \mu_{n-k+1} \ni \alpha.
\ee
\end{notation}
Consider an arbitrary unitary irrep $\phi ^{\mu }$ of $S(n)$. It
can be always unitarily transformed to $PRIR$ $\phi _{R}^{\mu },$ such that 
\be
\label{M1}
\forall \varkappa \in H\quad \phi _{R}^{\mu }(\varkappa )=\bigoplus _{\alpha
	\in \mu ,r_{\mu/\alpha}}\varphi ^{\alpha ,r_{\mu/\alpha}}(\varkappa )\equiv \bigoplus _{r_{\mu/\alpha}}\varphi ^{r_{\mu/\alpha}}(\varkappa ),
\ee
where $\alpha $ labels the type of if a irrep of $H$ and $r_{\mu/\alpha}$ denotes path on Young's lattice from $\mu$ to $\alpha$. It means that element $\varphi ^{\alpha ,r_{\mu/\alpha}}(\varkappa )$ is repeated $|\mathcal{R}_{\mu/\alpha}|=m_{\mu/\alpha}$ times, where $\mathcal{R}_{\mu/\alpha}$ is the set composed of all paths $r_{\mu/\alpha}$ from $\mu$ to $\alpha$. Whenever it is clear from the context we write just $\varphi ^{\alpha }(\varkappa )$  instead of $\varphi ^{\alpha ,r_{\mu/\alpha}}(\varkappa )$. Diagonal blocks in the decomposition~\eqref{M1} are labelled and in fact
ordered by the two indices $\alpha ,r_{\mu/\alpha}$. The $PRIR$
representation of $S(n)$, reduced to the subgroup $H$, has block diagonal form
of completely reduced representation, which in matrix notation takes the form
\be
\label{M2}
\forall \varkappa \in H\quad (\phi _{R}^{\mu })_{i_{\alpha }\quad 
	j_{\beta }}^{r_{\mu/\alpha}, \widetilde{r}_{\mu/\beta}}(\varkappa )=\delta
^{r_{\mu/\alpha} \widetilde{r}_{\mu/\beta}}\varphi _{i_{\alpha }j_{\alpha }}^{\alpha}(\varkappa),
\ee
where indices $i_{\alpha},j_{\alpha}$ run from 1 to dimension of the irrep $\alpha$, and $\delta
^{r_{\mu/\alpha} \widetilde{r}_{\mu/\beta}}=\delta^{\mu\nu}\delta^{\mu_{n-1}\nu_{n-1}}\cdots \delta^{\alpha\beta}$. The above considerations allow us to introduce the following 
\begin{notation}
	\label{not16}
	Every basis index $i_{\mu}$, where $\mu \vdash n$, can be written uniquely using a path on Young's lattice as
	\be
	\label{vv}
	i_{\mu}\equiv(r_{\mu/\alpha}, l_{\alpha}), \qquad     \alpha\in\mu,
	\ee
	and $l_{\alpha}$ denotes now index running only within the range of the irrep $\alpha$. The indices $i_{\mu},l_{\alpha}$ are of the same type as $r_{\mu/\alpha}$, but with trivial last element, i.e. a single box Young diagram. Equation~\eqref{vv} defines the division of the chosen path on Young's lattice from diagram $\mu$ to single box diagram, through a diagram $\alpha$. By writing $\delta_{i_{\mu}j_{\nu}}$, where $\mu \vdash n$ and $\alpha \vdash n-k$, we understand the following
	\be
	\delta_{i_{\mu}j_{\nu}}=\delta^{r_{\mu/\alpha}\widetilde{r}_{\nu/\beta}}\delta_{l_{\alpha}l'_{\beta}}=\delta_{\mu\nu}\delta_{\mu_{n-1}\nu_{n-1}}\cdots \delta_{\mu_{n-k+1}\nu_{n-k+1}}\delta_{\alpha\beta}\delta_{l_{\alpha}l'_{\beta}}.
	\ee
\end{notation}
Similarly as in~\cite{Studzinski2017,MozJPA} the block structure of this reduced representation
allows to introduce such a block indexation for the $PRIR$ $\phi_{R}^{\mu }$
of $S(n)$, which gives 
\be
\forall \sigma \in S(n)\quad \phi _{R}^{\mu }(\sigma )=\left( (\phi _{R}^{\mu })_{i_{\alpha }\quad j_{\beta }}^{r_{\mu/\alpha}, \widetilde{r}_{\mu/\beta}}(\sigma)\right) , 
\ee
where the matrices on the diagonal $(\phi _{R}^{\mu })_{i_{\alpha }\quad j_{\beta }}^{r_{\mu/\alpha}, \widetilde{r}_{\mu/\beta}}(\sigma)$ are of dimension of corresponding irrep $
\varphi ^{\alpha }$ of $S(n-1)$. The off diagonal blocks need not to be
square.

Now we formulate the  main result of this subsection, the generalized version of $PRIR$ orthogonality relation. The following proposition plays the central role in investigating matrix elements of MPBT operator in irreducible orthonormal operator basis presented later in Section~\ref{comm_structure}.
\begin{proposition}
	\label{summation0}
	Let $H\subset S(n)$ be an arbitrary subgroup of $S(n)$ with transversal $%
	T=\{\tau _{k}:k=1,\ldots,\frac{n!}{|H|}\}$,
	In the $PRIR$ notation $\phi _{R}^{\mu }$ of $S(n)$ satisfy the
	following bilinear sum rule%
	\begin{multline}
	\sum_{k=1}^{\frac{n!}{|H|}}\sum_{k_{\beta }=1}^{d_{\beta}}(\phi
	_{R}^{\mu })_{i_{\alpha }\quad  k_{\beta }}^{r_{\mu/\alpha},\widetilde{r}_{\mu/\beta}}(\tau _{k}^{-1})(\phi _{R}^{\nu })_{k_{\beta }\quad
		 j_{\gamma }}^{r_{\nu/\beta}, \widetilde{r}_{\nu/\gamma}}(\tau _{k})\\
	=\frac{n!}{|H|}\frac{d_{\beta }}{d_{\mu }}\delta ^{r_{\mu/\alpha}\widetilde{r}_{\nu/\gamma}}\delta _{i_{\alpha
		}j_{\gamma}}, \;
	\forall r_{\mu/\alpha},\widetilde{r}_{\mu/\beta}, r_{\nu/\beta},\widetilde{r}_{\nu/\gamma}
\end{multline}
	where $\alpha ,\beta ,\gamma $ are $irreps$ of $H$ contained in
	the irrep $\mu $ of $S(n)$, and $|H|$ denotes cardinality of the subgroup $H$.
\end{proposition}

\begin{proof}
	The proof is based on the classical orthogonality relations for irreps, which in PRIR notation takes a form
	\be
	\sum_{g\in G}(\phi _{R}^{\mu })_{i_{\alpha }\quad   k_{\beta }}^{r_{\mu/\alpha}\widetilde{r}_{\mu/\beta}}(g^{-1})(\phi _{R}^{\nu })_{k_{\beta }\quad
		 j_{\gamma }}^{r_{\nu/\beta}\widetilde{r}_{\nu/\gamma}}(g)=\frac{|G|}{%
		d_{\mu}}\delta ^{\widetilde{r}_{\mu/\beta}r_{\nu/\beta}}\delta _{i_{\alpha }j_{\gamma }},
	\ee
	where $|G|$ denotes cardianlity of the group $G$.
	It means, that even if $\alpha =\gamma$, i.e.  these representations
	are of the same type, but $\widetilde{r}_{\mu/\beta}\neq r_{\nu/\beta}$, the $RHS$ of the
	above equation is equal to zero. Next part of the proof follows from the proof of Proposition 29 in paper~\cite{MozJPA}.
\end{proof}

\subsection{Properties of irreducible operator basis and Young projectors under partial trace}
\label{parEij}
For further purposes, namely for effective computations of performance of our teleportation schemes, we prove here how irreducible operator basis given in~\eqref{Eij} or~\eqref{schur_Eij}, and Young projectors from~\eqref{def_P} behave under taking a partial trace over last $k$ systems. Our formulas are generalisations of attempts to similar problem made in~\cite{Aud}. We start considerations from calculating the partial trace from operators~\eqref{Eij} over last system. In all lemmas presented below we use PRIR representation described in Subsection~\ref{prir}.
\begin{lemma}
	\label{L3a}
	For irreducible operator basis $E^{\mu}_{kl}$, where $\mu \vdash n$, introduced in~\eqref{Eij}, the partial trace over last system equals to
	\be
	\tr_{n}E_{i_{\beta}j_{\beta'}}^{\beta \beta'}(\mu)=\sum_{\alpha \in\mu}\frac{m_{\mu}}{m_{\alpha}}E_{i_{\alpha}j_{\alpha}}^{\alpha}\delta_{\alpha \beta}\delta_{\alpha \beta'}=\frac{m_{\mu}}{m_{\beta}}E^{\beta}_{i_{\beta}j_{\beta}}\delta_{\beta \beta'}.
	\ee
\end{lemma}
\begin{proof}
	Similarly as it was done for Young projectors $P_{\mu}$ in~\eqref{P1}, we can rewrite $E^{\mu}_{ij}$ as
	\be
	\begin{split}
		E^{\mu}_{kl}&=\frac{d_{\mu}}{n!}\sum_{\sigma \in S(n)} \phi^{\mu}_{lk}(\sigma^{-1}) V_{\sigma}\\
								&=\frac{d_{\mu}}{n!}\sum_{a=1}^{n}\sum_{\tau \in S(n-1)}\phi^{\mu}_{lk}((a, n)\circ \tau^{-1})V_{a,n}V_{\tau}.
	\end{split}
	\ee
	Observing that $\phi^{\mu}_{lk}(\sigma^{-1})=\tr\left(|k\>\<l|\phi^{\mu}(\sigma^{-1}) \right)$, where $|k\>,|l\>$ are basis vector in irrep $\mu$, we can write in PRIR notation $k=k_{\mu}=(\beta,i_{\beta})$ and $l=l_{\mu}=(\beta',j_{\beta'})$ having
	\begin{multline}
	E_{i_{\beta}j_{\beta'}}^{\beta \beta'}(\mu)=
	\frac{d_{\mu}}{n!}\sum_{a=1}^nV_{a,n}\sum_{\tau\in S(n-1)}\\
	\tr\left[|\beta,i_{\beta}\>\<\beta',j_{\beta'}|\phi^{\mu}(a,n)\phi^{\mu}(\tau^{-1}) \right]V_{\tau}.
\end{multline}
	Since $\tau \in S(n-1)$ and $\mu \vdash n$, we can apply directly decomposition from~\eqref{M1} writing
	\begin{multline}
	E_{i_{\beta}j_{\beta'}}^{\beta \beta'}(\mu)=\frac{d_{\mu}}{n!}\sum_{a=1}^nV_{a,n}\sum_{\alpha\in\mu}\sum_{\tau\in S(n-1)}\\ \tr\left( \left[|\beta,i_{\beta}\>\<\beta',j_{\beta'}|\phi^{\mu}(a,n)\right]_{\alpha}\phi^{\alpha}(\tau^{-1})\right)  V_{\tau}.
\end{multline}
	In the above, by $\left[|\beta,i_{\beta}\>\<\beta',j_{\beta'}|\phi^{\mu}(a,n)\right]_{\alpha}$ we denote the restriction to irrep $\alpha$. Applying Fact~\ref{f1a} to the above expression we write 
	\be
	\begin{split}
		E_{i_{\beta}j_{\beta'}}^{\beta \beta'}&(\mu)\\
		=&\frac{d_{\mu}}{n!}\sum_{a=1}^{n}V_{a,n}\sum_{\alpha \in \mu}\frac{(n-1)!}{d_{\alpha}}\mathbf{1}^U_{\alpha}\ot \left[|\beta,i_{\beta}\>\<\beta',j_{\beta'}|\phi^{\mu}(a,n) \right]_{\alpha}.
	\end{split}
	\ee
Taking the partial trace over last system, and having in mind the definition of $\widetilde{\phi}^{\mu}(a,n-1)$ from~\eqref{obj}, we have:
	\be
	\begin{split}
		\tr_{n} & E_{i_{\beta}j_{\beta'}}^{\beta \beta'}(\mu)\\ & =\frac{d_{\mu}}{n!}\tr_{n}\left(\sum_{a=1}^{n}V_{a,n}\right)\sum_{\alpha \in \mu}\frac{(n-1)!}{d_{\alpha}}\mathbf{1}^U_{\alpha}\ot \left[|\beta,i_{\beta}\>\<\beta',j_{\beta'}|\phi^{\mu}(a,n) \right]_{\alpha} \\
		&=\frac{d_{\mu}}{n!}\sum_{a=1}^{n}\left(\sum_{\alpha \in \mu}\frac{(n-1)!}{d_{\alpha}}\mathbf{1}^U_{\alpha}\ot \left[|\beta,i_{\beta}\>\<\beta',j_{\beta'}|\widetilde{\phi}^{\mu}(a,n) \right]_{\alpha} \right).
	\end{split}
	\ee
	From the proof of Lemma~\ref{object} we know that the  object $\left[ \sum_{a=1}^{n}\widetilde{\phi}^{\mu}(a,n)\right]_{\beta}$ is invariant with respect to $S(n-1)$. Together with property $\mathbf{1}^{\alpha}_{\mu}|\beta,i_{\beta}\>\<\beta',j_{\beta'}|\mathbf{1}^{\alpha}_{\mu}=\delta_{\alpha\beta}\delta_{\alpha,\beta'}|\alpha,i_{\alpha}\>\<\alpha,j_{\alpha}|$, we have
	\be
	\begin{split}
		\tr_{n}&E_{i_{\beta}j_{\beta'}}^{\beta \beta'}(\mu)\\ & =\frac{d_{\mu}}{n!}\sum_{\alpha\in\mu}\frac{(n-1)!}{d_{\alpha}}x^{\mu}_{\alpha}\mathbf{1}^U_{\alpha}\ot |\alpha,i_{\alpha}\>\<\alpha,j_{\alpha}|\delta_{\alpha\beta}\delta_{\alpha\beta'}\\ 
					 & =\frac{1}{n}\frac{d_{\mu}}{d_{\beta}}x^{\mu}_{\beta}\mathbf{1}^U_{\alpha}\ot |\beta,i_{\beta}\>\<\beta,j_{\beta}|\delta_{\beta\beta'}=\frac{m_{\mu}}{m_{\beta}}E^{\beta}_{i_{\beta}j_{\beta}}\delta_{\beta\beta'}.
	\end{split}
	\ee
	In the last step we use explicit form of coefficients $x^{\mu}_{\beta}$ given in Lemma~\ref{object} and expression~\eqref{schur_Eij}.
\end{proof}

\begin{corollary}
	From Lemma~\ref{L3} we see that taking a partial trace over $n-$th subsystem we destroys all the coherences between block labelled by different $\beta\vdash n-1$.
\end{corollary}
 For further purpose of having explicit connection with the structure of multi-port teleportation scheme, let us assume that now $\mu \vdash n-k$, such that $2k<n$. Having that and extended notion of $PRIR$, we are in position to present the second main result of this section.
\begin{lemma}
\label{L3}
For basis operators $E^{\mu}_{kl}$ in the irreducible representation labelled by $\mu\vdash n-k$, we have the following equality:
\be
\tr_{(k)}E_{i_{\beta}\quad j_{\beta'}}^{r_{\mu/\beta} \  \widetilde{r}_{\mu/\beta'}}=\frac{m_{\mu}}{m_{\beta}}E^{\beta}_{i_{\beta}j_{\beta}}\delta_{r_{\mu/\beta}  \widetilde{r}_{\mu/\beta'}}
\ee
where we use simplified notation $\tr_{(k)}=\tr_{n-2k+1,\ldots,n-k}$.
\end{lemma}

\begin{proof}
To prove the above statement we use iteratively Lemma~\ref{L3a}. Let us write explicitly indices $k=k_{\mu}, l=l_{\mu}$ in PRIR notation:
\be
\label{paths}
\begin{split}
k_{\mu}&=(\mu_{n-k-1},r_{\mu_{n-k-1}})=(\mu_{n-k-1},\ldots,\mu_{n-2k+1},\mu_{n-2k},i_{\mu_{n-2k}})\\ &=(\mu_{n-k-1},\ldots,\mu_{n-2k+1},\beta,i_{\beta}),\\
l_{\mu}&=(\mu'_{n-k-1},s_{\mu'_{n-k-1}})=(\mu'_{n-k-1},\ldots,\mu'_{n-2k+1},\mu'_{n-2k},i_{\mu'_{n-2k}})\\ &=(\mu'_{n-k-1},\ldots,\mu'_{n-2k+1},\beta',j_{\beta'}),
\end{split}
\ee 
where we put $\beta=\mu_{n-2k},\beta'=\mu'_{n-2k}$ for simpler notation. Each lower index denotes a proper layer on the reduced Young's lattice, starting from the highest layer labelled by the number $n-k$. In the first step we compute the partial trace over $(n-k)$-th system getting
\be
\label{n-k}
\begin{split}
	\tr_{n-k}&E_{r_{\mu_{n-k-1}}s_{\mu'_{n-k-1}}}^{\mu_{n-k-1} \  \mu'_{n-k-1}}(\mu_{n-k})\\ &=\frac{m_{\mu_{n-k}}}{m_{\mu_{n-k-1}}}E^{\mu_{n-k-1}}_{r_{\mu_{n-k-1}}s_{\mu_{n-k-1}}}\delta_{\mu_{n-k-1}\mu'_{n-k-1}}.
\end{split}
\ee
This procedure reduced paths in~\eqref{paths} to
\be
\begin{split}
	r_{\mu_{n-k-1}}&=(\mu_{n-k-2},q_{\mu_{n-k-2}})\\&=(\mu_{n-k-2},\mu_{n-k-3},\ldots,\mu_{n-2k+1},\beta,i_{\beta}),\\ \end{split}
\ee
\be
\begin{split}
	s_{\mu_{n-k-1}}&=(\mu'_{n-k-2},p_{\mu'_{n-k-2}})\\ &=(\mu'_{n-k-2},\mu'_{n-k-3},\ldots,\mu'_{n-2k+1},\beta',j_{\beta'}),
\end{split}
\ee
where $\beta=\mu_{n-2k},\beta'=\mu'_{\mu-2k}$.
Now computing the trace from~\eqref{n-k} over $(n-k-1)$-th particle we write
\be
\begin{split}
&\delta_{\mu_{n-k-1}\mu'_{n-k-1}}\frac{m_{\mu_{n-k}}}{m_{\mu_{n-k-1}}}\tr_{n-k-1}E^{\mu_{n-k-2}\mu'_{n-k-2}}_{q_{\mu_{n-k-2}}p_{\mu'_{n-k-2}}}(\mu_{n-k-1})\\
&=\delta_{\mu_{n-k-1}\mu'_{n-k-1}}\delta_{\mu_{n-k-2}\mu'_{n-k-2}}\frac{m_{\mu_{n-k}}}{m_{\mu_{n-k-1}}}\frac{m_{\mu_{n-k-1}}}{m_{\mu_{n-k-2}}}E^{\mu_{n-k-2}}_{q_{\mu_{n-k-2}}p_{\mu_{n-k-2}}}.
\end{split}
\ee
Continuing the above procedure, up to last system in $\tr_{(k)}$ we obtain expression~\eqref{eq:long2}, displayed at the top of the following page
since in the last line we used definition of $r_{\mu/\alpha}$ and suppressed indices labelling layers on reduced Bratelli diagram. This finishes the proof.
\end{proof}

\begin{figure*}[!t]
	\normalsize
	\setcounter{MYtempeqncnt}{\value{equation}}
	\setcounter{equation}{79}
\be
\begin{split}\label{eq:long2}
&\delta_{\mu_{n-k-1}\mu'_{n-k-1}}\delta_{\mu_{n-k-2}\mu'_{n-k-2}}\times \cdots \times \delta_{\mu_{n-2k+1}\mu'_{n-2k+1}}\delta_{\mu_{n-2k}\mu'_{n-2k}}\frac{m_{\mu_{n-k}}}{m_{\mu_{n-k-1}}}\frac{m_{\mu_{n-k-1}}}{m_{\mu_{n-k-2}}}\times \cdots \times \frac{m_{\mu_{n-2k+1}}}{m_{\mu_{n-2k}}}E^{\mu_{n-2k}}_{i_{\mu_{n-2k}}j_{\mu_{n-2k}}}\\
&=\delta_{\mu_{n-k-1}\mu'_{n-k-1}}\delta_{\mu_{n-k-2}\mu'_{n-k-2}}\times \cdots \times \delta_{\mu_{n-2k+1}\mu'_{n-2k+1}}\delta_{\mu_{n-2k}\mu'_{n-2k}}\frac{m_{\mu_{n-k}}}{m_{\mu_{n-2k}}}E^{\mu_{n-2k}}_{i_{\mu_{n-2k}}j_{\mu_{n-2k}}}=\delta^{r_{\mu/\beta}\widetilde{r}_{\mu/\beta'}}\frac{m_{\mu}}{m_{\beta}}E^{\beta}_{i_{\beta}j_{\beta}},
\end{split}
\ee
\setcounter{equation}{\value{MYtempeqncnt}+1}
\hrulefill
\vspace*{4pt}
\end{figure*}

 Then Lemma~\ref{L3} implies the following statement about the Young projector:
\begin{corollary}
	\label{corL3}
	Let $P_{\mu}$ be a Young projector on irrep labelled by $\mu \vdash n-k$, then
	\be
	\tr_{(k)}P_{\mu}=\sum_{\beta \in \mu}m_{\mu/\beta}\frac{m_{\mu}}{m_{\beta}}P_{\beta}
	\ee
	where we use simplified notation $\tr_{(k)}=\tr_{n-2k+1,\ldots,n-k}$.
\end{corollary}
Indeed, knowing that $P_{\mu}=\sum_{k}E^{\mu}_{ii}$, we write in PRIR basis
\be
\begin{split}
	\tr_{(k)}P_{\mu}&=\sum_{k_{\mu}}E^{\mu}_{k_{\mu}k_{\mu}}=\sum_{\beta \in \mu}\sum_{r_{\mu/\beta}}\sum_{i_{\beta}}\tr_{(k)}E_{i_{\beta}\quad i_{\beta}}^{r_{\mu/\beta}  r_{\mu/\beta}}\\ & =\sum_{\beta \in \mu}\sum_{r_{\mu/\beta}}\sum_{i_{\beta}}\frac{m_{\mu}}{m_{\beta}}E^{\beta}_{i_{\beta}i_{\beta}}
	=\sum_{\beta \in \mu}\sum_{r_{\mu/\beta}}\sum_{i_{\beta}}\frac{m_{\mu}}{m_{\beta}}E^{\beta}_{i_{\beta}i_{\beta}}\\ & =\sum_{\beta \in \mu}\sum_{r_{\mu/\beta}}\frac{m_{\mu}}{m_{\beta}}P_{\beta}=\sum_{\beta \in \mu}m_{\mu/\beta}\frac{m_{\mu}}{m_{\beta}}P_{\beta}.
\end{split}
\ee
\section{The Commutant Structure of $U^{\ot (n-k)}\ot \overline{U}^{\ot k}$ Transformations and MPBT operator}
\label{comm_structure}
 In this section we deliver an orthonormal basis for the commutant of $U^{\ot (n-k)}\ot \overline{U}^{\ot k}$, or equivalently for the algebra $\mathcal{A}^{(k)}_n(d)$. Being more strict, we introduce an irreducible basis for an two-sided ideal $\mathcal{M}$ generated by the element $V^{(k)}$ and elements of the algebra $\mathcal{A}^{(k)}_n(d)$:
\be
\label{idealM}
\mathcal{M}=\{V_{\tau}V^{(k)}V^{\dagger}_{\tau'} \ | \ \tau,\tau'\in S(n-k)\}.
\ee
For our problem full description of $\mathcal{M}$, together with irreducible representation is enough since all basic objects describing MPBT scheme belong to this ideal, see for example definition of MPBT operator from~\eqref{PBT1}. In the most general case the algebra $\mathcal{A}^{(k)}_n(d)$ contains also two-sided ideals generated by the elements $V^{(k')}$, for $k'<k$, and elements of the algebra $\mathcal{A}^{(k)}_n(d)$. We have the following chain of inclusions
\be
\mathcal{M}\equiv \mathcal{M}^{(k)}\subset \mathcal{M}^{(k-1)}\subset \cdots \subset \mathcal{M}^{(1)}\subset \mathcal{M}^{(0)}\equiv \mathcal{A}^{(k)}_n(d).
\ee
The irreducible basis fir the ideals with $k'< k$ will be studied elsewhere, since we do not use objects from the outside of the ideal $\mathcal{M}$. In Figure~\ref{structure_M} we present nested structure of $\mathcal{A}^{(2)}_5(d)$ for $d>3$, together with labelling subsequent blocks within them.
\begin{figure}[h]
	\includegraphics[width=\linewidth]{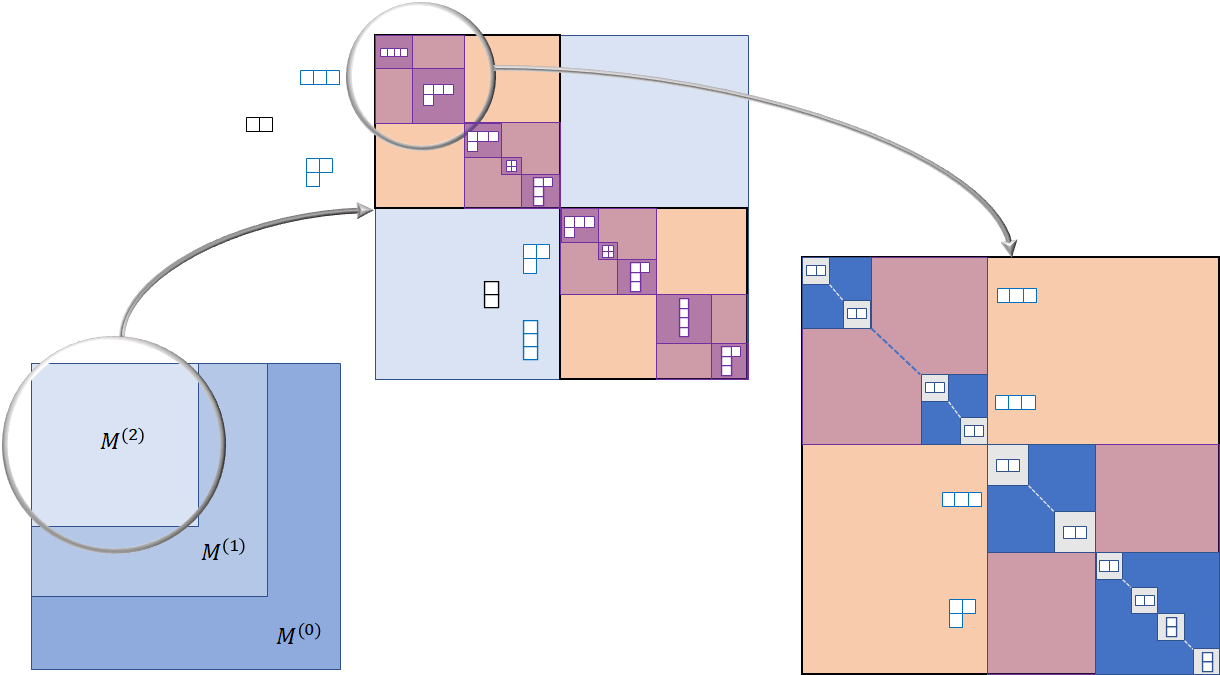}
	\caption{Graphic presents the interior structure of the algebra $\mathcal{A}^{(2)}_6(d)$, with the nested structure of the ideals $\mathcal{M}^{(0)}, \mathcal{M}^{(1)}, \mathcal{M}^{(2)}$, for $d\geq4$ (only with this requirement all the Young frames occur in the decomposition). In particular, we focus on the interior block structure of the ideal $\mathcal{M}^{(2)}$, on which objects describing multi-port teleportation schemes are defined. The middle figure represents the process of the induction by adding two boxes to two allowed starting Young frames which are $(2)$ and $(1,1)$. In this case we have nested structure of three layers.  The most right figure presents process of the reduction from irreps labelled by Young frames of 4 boxes to irreps labelled to Young frames of two boxes. We present here the process of the reduction for the two most left upper blocks.}
	\label{structure_M}
\end{figure}
Having expressions for partial trace over an arbitrary number of particles from irreducible basis operators of the symmetric group we are in the position to formulate the main result, namely we have:
\begin{theorem}	
	\label{tmbas}
	The orthonormal operator basis of the commutant of $U^{\ot (n-k)}\ot \overline{U}^{\ot k}$ in the maximal ideal $\mathcal{M}$ is given by the following set of operators
	\be
	\label{tmbas2}
	F^{r_{\mu/\alpha} r_{\nu/\alpha}}_{i_{\mu}\quad j_{\nu}}=\frac{m_{\alpha}}{\sqrt{m_{\mu}m_{\nu}}}E_{i_{\mu} \ 1_{\alpha}}^{\quad r_{\mu/\alpha}}V^{(k)}E^{r_{\nu/\alpha}}_{1_{\alpha}\quad  j_{\nu}}
	\ee
	satisfying the following composition rule
	   \be
	\label{orto}
	F^{r_{\mu/\alpha} r_{\nu/\alpha}}_{i_{\mu}\quad j_{\nu}}F^{r_{\mu'/\beta} r_{\nu'/\beta}}_{k_{\mu'}\quad l_{\nu'}}=\delta^{r_{\nu/\alpha}r_{\mu'/\beta}}\delta_{j_{\nu}k_{\mu'}}F^{r_{\mu/\alpha} r_{\nu'/\alpha}}_{i_{\mu}\quad l_{\nu'}}
	\ee
	where $m_{\mu},m_{\nu}$ and $m_{\alpha}$ are multiplicities of respective irreps of $S(n-k)$ and $S(n-2k)$ in the Schur-Weyl duality.
\end{theorem}

\begin{proof}
The proof contains two main steps:
\begin{itemize}
	\item Showing that operators are orthonormal, i.e.
	\be
	\label{orto2}
	F^{r_{\mu/\alpha} r_{\nu/\alpha}}_{i_{\mu}\quad j_{\nu}}F^{r_{\mu'/\beta} r_{\nu'/\beta}}_{k_{\mu'}\quad l_{\nu'}}=\delta^{r_{\nu/\alpha}r_{\mu'/\beta}}\delta_{j_{\nu}k_{\mu'}}F^{r_{\mu/\alpha} r_{\nu'/\alpha}}_{i_{\mu}\quad l_{\nu'}}.
	\ee
	Indeed, writing explicitly the above composition and using orthogonality relation for operators $E^{\mu}_{i_{\mu}j_{\mu}}$, we have
	
	\begin{IEEEeqnarray}{lll}
		& F^{r_{\mu/\alpha} r_{\nu/\alpha}}_{i_{\mu}\quad j_{\nu}}F^{r_{\mu'/\beta} r_{\nu'/\beta}}_{k_{\mu'}\quad l_{\nu'}}  & \nonumber \\
		& = \frac{m_{\alpha}}{\sqrt{m_{\mu}m_{\nu}}}\frac{m_{\beta}}{\sqrt{m_{\mu'}m_{\nu'}}} & E_{i_{\mu} \ 1_{\alpha}}^{\quad r_{\mu/\alpha}}V^{(k)}E^{r_{\nu/\alpha}}_{1_{\alpha}\quad  j_{\nu}}  \nonumber\\
		& & \times		E_{k_{\mu'} \ 1_{\beta}}^{\quad r_{\mu'/\beta}}V^{(k)}E^{r_{\nu'/\beta}}_{1_{\beta}\quad  l_{\nu'}} \nonumber\\
		&=\delta^{\nu\mu'}\delta_{j_{\nu}k_{\mu'}}\frac{m_{\alpha}}{\sqrt{m_{\mu}m_{\nu}}} & \frac {m_{\beta}}{\sqrt{m_{\mu'}m_{\nu'}}}E_{i_{\mu} \ 1_{\alpha}}^{\quad r_{\mu/\alpha}}V^{(k)}  \nonumber \\
		&  &  \times E_{1_{\alpha}\quad1_{\beta}}^{r_{\nu/\alpha}r_{\mu'/\beta}}V^{(k)}E^{r_{\nu'/\beta}}_{1_{\beta}\quad l_{\nu'}}.
	\end{IEEEeqnarray}
	
	Now, applying Fact~\ref{kpartr} to operator $E_{1_{\alpha}\quad1_{\beta}}^{r_{\nu/\alpha}r_{\nu'/\beta}}$, together with Lemma~\ref{L3}, we reduce to
	\begin{IEEEeqnarray}{rll}
	 & F^{r_{\mu/\alpha} r_{\nu/\alpha}}_{i_{\mu}\quad j_{\nu}} F^{r_{\mu'/\beta} r_{\nu'/\beta}}_{k_{\mu'}\quad l_{\nu'}}& \nonumber\\
	 & =\delta^{\nu\mu'}\delta^{\alpha \beta}\delta^{r_{\nu/\alpha}r_{\mu'/\beta}} & \delta_{j_{\nu}k_{\mu'}}\frac{m_{\alpha}}{\sqrt{m_{\mu}m_{\mu'}}}\frac{m_{\alpha}}{\sqrt{m_{\mu'}m_{\nu'}}}\frac{m_{\mu'}}{m_{\alpha}} \nonumber\\ 
	 & & \times E_{i_{\mu} \ 1_{\alpha}}^{\quad r_{\mu/\alpha}}E^{\alpha}_{1_{\alpha}1_{\alpha}}V^{(k)}E^{r_{\nu'/\alpha}}_{1_{\alpha}\quad l_{\nu'}}\nonumber\\
	 & =\delta^{r_{\nu/\alpha}r_{\mu'/\beta}}\delta_{j_{\nu}k_{\mu'}} & \frac{m_{\alpha}}{\sqrt{m_{\mu}m_{\nu'}}} \nonumber\\
	 & & \times E_{i_{\mu} \ 1_{\alpha}}^{\quad r_{\mu/\alpha}}E^{\alpha}_{1_{\alpha}1_{\alpha}}V^{(k)}E^{r_{\nu'/\alpha}}_{1_{\alpha}\quad l_{\nu'}}.
\end{IEEEeqnarray}
	Finally observing that $E_{i_{\mu} \ 1_{\alpha}}^{\quad r_{\mu/\alpha}}E^{\alpha}_{1_{\alpha}1_{\alpha}}=E_{i_{\mu} \ 1_{\alpha}}^{\quad r_{\mu/\alpha}}$ we get expression~\eqref{orto}.
	\item Showing that element $V^{(k)}$ generating the ideal $\mathcal{M}$, see~\eqref{idealM} can be expressed as a linear combination of basis elements $F^{r_{\mu/\alpha} r_{\nu/\alpha}}_{i_{\mu}\quad j_{\nu}}$. Indeed, we have
	\be
	\begin{split}
	V^{(k)}=\sum_{\mu,\nu \vdash n-k} P_{\mu}V^{(k)}P_{\nu}=\sum_{\mu,\nu \vdash n-k} \ \sum_{i_{\mu},j_{\nu}}E^{\mu}_{i_{\mu}i_{\mu}}V^{(k)}E^{\nu}_{j_{\nu}j_{\nu}},
	\end{split}
	\ee
	since $\mathbf{1}=\sum_{\mu}P_{\mu}$ together with~\eqref{def_P}. Writing indices $i_{\mu},j_{\nu}$ in PRIR notation, according to Notation~\ref{not16} we get
		\be
		\label{to}
		\begin{split}
			V^{(k)}& =\sum_{\mu,\nu}\sum_{r_{\mu/\alpha},\widetilde{r}_{\nu/\beta}}\sum_{l_{\alpha},l'_{\beta}}E^{r_{\mu/\alpha}r_{\mu/\alpha}}_{l_{\alpha}\quad l_{\alpha}}V^{(k)}E^{\widetilde{r}_{\nu/\beta}\widetilde{r}_{\nu/\beta}}_{l'_{\beta}\quad l'_{\beta}}\\
			       &=\sum_{\mu,\nu}\sum_{r_{\mu/\alpha},\widetilde{r}_{\nu/\beta}}\sum_{l_{\alpha},l'_{\beta}}E^{r_{\mu/\alpha}r_{\mu/\alpha}}_{l_{\alpha}\quad 1_{\alpha}}E^{\alpha}_{1_{\alpha}l_{\alpha}}V^{(k)}E^{\beta}_{l'_{\beta}1_{\beta}}E^{\widetilde{r}_{\nu/\beta}\widetilde{r}_{\nu/\beta}}_{1_{\beta}\quad l'_{\beta}}.
		\end{split}	
		\ee
		Having $[E^{\alpha}_{1_{\alpha}l_{\alpha}},V^{(k)}]=[E^{\beta}_{l'_{\beta}1_{\beta}},V^{(k)}]=0$ and orthogonality relation $E^{\alpha}_{1_{\alpha}l_{\alpha}}E^{\beta}_{l'_{\beta}1_{\beta}}=\delta^{\alpha\beta}\delta_{l_{\alpha}l'_{\beta}}E^{\alpha}_{1_{\alpha}1_{\alpha}}$ we reduce~\eqref{to} to
		\be
		\begin{split}
		V^{(k)}&=\sum_{\mu,\nu} \ \sum_{r_{\mu/\alpha},\widetilde{r}_{\nu/\alpha}}\sum_{l_{\alpha}}E^{r_{\mu/\alpha}r_{\mu/\alpha}}_{l_{\alpha}\quad 1_{\alpha}}V^{(k)}E^{\widetilde{r}_{\nu/\alpha}\widetilde{r}_{\nu/\alpha}}_{1_{\alpha}\quad l_{\alpha}}\\
		&=  \sum_{\mu,\nu} \ \sum_{r_{\mu/\alpha},\widetilde{r}_{\nu/\alpha}}\sum_{l_{\alpha}}\frac{\sqrt{m_{\mu}m_{\nu}}}{m_{\alpha}}\left(\frac{m_{\alpha}}{\sqrt{m_{\mu}m_{\nu}}} E^{r_{\mu/\alpha}r_{\mu/\alpha}}_{l_{\alpha}\quad 1_{\alpha}}V^{(k)}E^{\widetilde{r}_{\nu/\alpha}\widetilde{r}_{\nu/\alpha}}_{1_{\alpha}\quad l_{\alpha}}\right)\\
		&=\sum_{\mu,\nu} \ \sum_{r_{\mu/\alpha},\widetilde{r}_{\nu/\alpha}}\sum_{l_{\alpha}}\frac{\sqrt{m_{\mu}m_{\nu}}}{m_{\alpha}} F^{\quad r_{\mu/\alpha} \   \widetilde{r}_{\nu/\alpha}}_{l_{\alpha} \ r_{\mu/\alpha}  \ \widetilde{r}_{\nu/\alpha} \ l_{\alpha}}.
		\end{split}
		\ee
		In the above we use representation of $F^{r_{\mu/\alpha} r_{\nu/\alpha}}_{i_{\mu}\quad j_{\nu}}$ in full PRIR basis:
		\be
		\label{full}
		F^{r_{\mu/\alpha} r_{\nu/\alpha}}_{i_{\mu} \quad j_{\nu}} \rightarrow F^{\quad r_{\mu/\alpha} \   r_{\nu/\alpha}}_{k_{\beta} \ r_{\mu/\beta}  \ r_{\nu/\gamma} \ k'_{\gamma}}
		\ee
		since $i_{\mu}=(r_{\mu/\beta},k_{\beta})$ and $j_{\nu}=(r_{\nu/\gamma},k'_{\gamma})$. This finishes the proof.
		\end{itemize}
	\end{proof}
Next we focus on the relations analogous to~\eqref{actionE} for the basis elements $F^{r_{\mu/\alpha} r_{\nu/\alpha}}_{i_{\mu} \quad j_{\nu}}$ and operators $V^{(k)}$, $V_{\tau}$, where $\tau \in  \mathcal{S}_{n,k}\equiv \frac{S(n-k)}{S(n-2k)}$. To have all required tools let us first rewrite expressions from~\eqref{actionE} in PRIR notation, but for a specific choice of indices and partitions $\mu \vdash n-k$ and $\alpha\vdash n-2k$:
\be
\label{actionEE}
\begin{split}
&\forall \tau\in S(n-k)\qquad V_{\tau} E_{i_{\mu} \ 1_{\alpha}}^{\quad r_{\mu/\alpha}}=\sum_{l_{\mu}}\phi^{\mu}_{l_{\mu}i_{\mu}}(\tau)E_{l_{\mu} \ 1_{\alpha}}^{\quad r_{\mu/\alpha}}\\
&\forall \tau\in S(n-k)\qquad E^{r_{\nu/\alpha}}_{1_{\alpha}\quad j_{\nu}}V_{\tau^{-1}}=\sum_{k_{\nu}}\phi_{j_{\nu}k_{\nu}}^{\nu}(\tau^{-1})E_{1_{\alpha}\quad k_{\nu}}^{r_{\nu/\alpha}}
\end{split}
\ee
where $\phi^{\mu}_{l_{\mu}i_{\mu}}(\tau), \phi_{j_{\nu}k_{\nu}}^{\nu}(\tau^{-1})$ are the matrix elements of $V_{\tau}, V_{\tau^{-1}}$ in irreducible basis expressed in the PRIR notation, see~\eqref{blaa} and Section~\ref{preliminary}.
Having the above we are in position to prove the following
\begin{lemma}
\label{actionE'}
Let us take basis operators for the ideal $\mathcal{M}$ given through Theorem~\ref{tmbas}, together with~\eqref{full}. Then for the operator $V^{(k)}$ defined in~\eqref{parV} and an arbitrary permutation operator $V_{\tau}$, for $\tau \in S(n-k)$, the following relations hold:
\be
\label{eqactionE'}
F^{\quad r_{\mu/\alpha} \   r_{\nu/\alpha}}_{k_{\beta} \ r_{\mu/\beta}  \ r_{\nu/\gamma} \ l_{\gamma}}V^{(k)}=\sum_{\mu'}\sum_{r_{\mu'/\gamma}} \frac{\sqrt{m_{\nu}m_{\mu'}}}{m_{\gamma}} F^{\quad r_{\mu/\gamma} \   r_{\mu'/\gamma}}_{k_{\beta} \ r_{\mu/\beta}  \ r_{\mu'/\gamma} \ l_{\gamma}}\delta^{r_{\nu/\alpha}r_{\nu/\gamma}}
\ee
and
\be
\label{eqactionE'1}
F^{\quad r_{\mu/\alpha} \   r_{\nu/\alpha}}_{k_{\beta} \ r_{\mu/\beta}  \ r_{\nu/\gamma} \ l_{\gamma}}V_{\tau}=\sum_{k_{\nu}}\phi^{\nu}_{j_{\nu}k_{\nu}}(\tau) F^{r_{\mu/\alpha} r_{\nu/\alpha}}_{i_{\mu} \quad k_{\nu}}
\ee
where $\phi^{\nu}_{j_{\nu}k_{\nu}}(\tau) $ are the matrix elements of $V_{\tau}$ in the irreducible basis expressed in the PRIR notation introduced in Section~\ref{preliminary}.
\end{lemma}
\begin{proof}
 First let us calculate action of $F^{r_{\mu/\alpha} r_{\nu/\alpha}}_{i_{\mu}\quad j_{\nu}}$ on $V^{(k)}$.  Using expression~\eqref{full} we have
	\be
	\label{x1}
	\begin{split}
		  F^{\quad r_{\mu/\alpha} \   r_{\nu/\alpha}}_{k_{\beta} \ r_{\mu/\beta}  \ r_{\nu/\gamma} \ l_{\gamma}}V^{(k)} &=\frac{m_{\alpha}}{\sqrt{m_{\mu}m_{\nu}}}E^{r_{\mu/\beta}r_{\mu/\alpha}}_{k_{\beta}\quad 1_{\alpha}}V^{(k)}E^{r_{\nu/\alpha}r_{\nu/\gamma}}_{1_{\alpha}\quad l_{\gamma}}V^{(k)}\\ 
		& =\sqrt{\frac{m_{\nu}}{m_{\mu}}}\delta^{r_{\nu/\alpha}r_{\nu/\gamma}}E^{r_{\mu/\beta}r_{\mu/\gamma}}_{k_{\beta}\quad 1_{\gamma}}E^{\gamma}_{1_{\gamma}l_{\gamma}}V^{(k)},
	\end{split}
	\ee
	where in the second equality we used Fact~\ref{kpartr} and Lemma~\ref{L3}. Now decomposing identity acting on $n-k$ systems in PRIR basis
	\be
	\mathbf{1}=\sum_{\mu'\vdash n-k}P_{\mu'}=\sum_{\mu'}\sum_{r_{\mu'/\alpha'}}\sum_{s_{\alpha'}}E_{s_{\alpha'}\quad s_{\alpha'}}^{r_{\mu'/\alpha'}  r_{\mu'/\alpha'}}\qquad \alpha'\vdash n-2k,
	\ee
	and multiplying by it the right hand side of~\eqref{x1} we have
	\be
	\label{x2}
	\begin{split}
& F^{\quad r_{\mu/\alpha} \   r_{\nu/\alpha}}_{k_{\beta} \ r_{\mu/\beta}  \ r_{\nu/\gamma} \ l_{\gamma}}V^{(k)}\\ 
& =\sqrt{\frac{m_{\nu}}{m_{\mu}}}\sum_{\mu'}\sum_{r_{\mu'/\alpha'}}\sum_{s_{\alpha'}}E^{r_{\mu/\beta}r_{\mu/\gamma}}_{k_{\beta}\quad 1_{\gamma}}V^{(k)}E^{\gamma}_{1_{\gamma}l_{\gamma}}E_{s_{\alpha'}\quad s_{\alpha'}}^{r_{\mu'/\alpha'}  r_{\mu'/\alpha'}}\delta^{r_{\nu/\alpha}r_{\nu/\gamma}}
	\end{split}
	\ee
	since $\left[E^{\gamma}_{1_{\gamma}l_{\gamma}},V^{(k)}\right]=0$. Moreover we have $E^{\gamma}_{1_{\gamma}l_{\gamma}}E_{s_{\alpha'}\quad s_{\alpha'}}^{r_{\mu'/\alpha'}  r_{\mu'/\alpha'}}=\delta^{\alpha'\gamma}\delta_{s_{\alpha'}l_{\gamma}}E_{1_{\gamma}\quad l_{\gamma}}^{r_{\mu'/\gamma}   r_{\mu'/\gamma}}$. Substituting to~\eqref{x2} we write:
	\be
	\begin{split}
	& F^{\quad r_{\mu/\alpha} \   r_{\nu/\alpha}}_{k_{\beta} \ r_{\mu/\beta}  \ r_{\nu/\gamma} \ l_{\gamma}}V^{(k)} \\
	& = \sqrt{\frac{m_{\nu}}{m_{\mu}}}\sum_{\mu'}\sum_{r_{\mu'/\gamma}}E^{r_{\mu/\beta}r_{\mu/\gamma}}_{k_{\beta}\quad 1_{\gamma}}V^{(k)}E^{r_{\mu'/\gamma}r_{\mu'/\gamma}}_{1_{\gamma}\quad l_{\gamma}}\delta^{r_{\nu/\alpha}r_{\nu/\gamma}}\\ 
	& = \sum_{\mu'}\sum_{r_{\mu'/\gamma}} \frac{\sqrt{m_{\nu}m_{\mu'}}}{m_{\gamma}} F^{\quad r_{\mu/\gamma} \   r_{\mu'/\gamma}}_{k_{\beta} \ r_{\mu/\beta}  \ r_{\mu'/\gamma} \ l_{\gamma}}\delta^{r_{\nu/\alpha}r_{\nu/\gamma}}.
	\end{split}
	\ee
This proves expression~\eqref{eqactionE'}.To prove equation~\eqref{eqactionE'1} we use directly~\eqref{actionEE} with~\eqref{tmbas2}:
\be
\begin{split}
F^{r_{\mu/\alpha} r_{\nu/\alpha}}_{i_{\mu} \quad j_{\nu}}V_{\tau}& = \frac{m_{\alpha}}{\sqrt{m_{\mu}m_{\nu}}}E_{i_{\mu} \ 1_{\alpha}}^{\quad r_{\mu/\alpha}}V^{(k)}E^{r_{\nu/\alpha}}_{1_{\alpha}\quad  j_{\nu}}V_{\tau}\\
																																 & = \sum_{k_{\nu}}\phi_{j_{\nu}k_{\nu}}^{\nu}(\tau)\frac{m_{\alpha}}{\sqrt{m_{\mu}m_{\nu}}}E_{i_{\mu} \ 1_{\alpha}}^{\quad r_{\mu/\alpha}}V^{(k)}E^{r_{\nu/\alpha}}_{1_{\alpha}\quad  k_{\nu}}\\
																																 &=\sum_{k_{\nu}}\phi_{j_{\nu}k_{\nu}}^{\nu}(\tau)F^{r_{\mu/\alpha} r_{\nu/\alpha}}_{i_{\mu} \quad k_{\nu}}.
\end{split}
\ee
 This finishes the proof.
\end{proof}
Analogously we can evaluate expressions~\eqref{eqactionE'}, \eqref{eqactionE'1} for action from the right-hand side. For the further purposes we write explicitly such action on $V_{\tau}$, for $\tau \in S(n-k)$:
\be
\label{rhs_action}
V_{\tau}F^{r_{\mu/\alpha} r_{\nu/\alpha}}_{i_{\mu} \quad j_{\nu}}=\sum_{k_{\mu}}\phi_{k_{\mu}i_{\mu}}^{\mu}(\tau)F^{r_{\mu/\alpha} r_{\nu/\alpha}}_{k_{\mu} \quad j_{\nu}}.
\ee
Using the second part of the proof of Theorem~\ref{tmbas} we can formulate the following
\begin{lemma}
\label{Vel}
The operator $V^{(k)}$ defined in~\eqref{parV} and an arbitrary permutation operator $V_{\tau}$, for $\tau \in S(n-k)$ in the operator basis from Theorem~\ref{tmbas} have matrix elements equal to:
\be
\label{eqVel}
\left(V^{(k)}\right)^{\quad r_{\mu/\alpha} \   r_{\nu/\alpha}}_{k_{\beta} \ r_{\mu/\beta}  \ r_{\nu/\gamma} \ l_{\gamma}}=\delta_{k_{\beta}l_{\gamma}}\delta^{r_{\mu/\alpha}r_{\mu/\beta}}\delta^{r_{\nu/\alpha}r_{\nu/\gamma}}\frac{\sqrt{m_{\mu}m_{\nu}}}{m_{\alpha}},
\ee
and
\be
\label{eqVel2}
\left( V_{\tau}\right)^{r_{\mu/\alpha}r_{\nu/\alpha}}_{i_{\mu}\quad j_{\nu}}=\delta^{r_{\mu/\alpha}r_{\nu/\alpha}}\delta_{i_{\mu}j_{\nu}}\sqrt{\frac{m_{\mu}}{m_{\nu}}}\sum_{k_{\mu}}\phi^{\mu}_{k_{\mu}i_{\mu}}(\tau),
\ee
where $m_{\mu},m_{\nu},m_{\alpha}$ are multiplicities of respective irreducible representations in the Schur-Weyl duality, and $\phi^{\mu}_{k_{\mu}i_{\mu}}(\tau)$  are the matrix elements of $V_{\tau}$ in the irreducible basis expressed in the PRIR notation introduced in Section~\ref{preliminary}.
\end{lemma}

\begin{proof}
To prove the statement of the lemma we have to compute overlap of $V^{(k)}$ with $F^{r_{\mu/\alpha} r_{\nu/\alpha}}_{i_{\mu} \quad j_{\nu}}$ written in PRIR basis:
	\be
	\begin{split}
		\left(V^{(k)}\right)^{\quad r_{\mu/\alpha} \   r_{\nu/\alpha}}_{k_{\beta} \ r_{\mu/\beta}  \ r_{\nu/\gamma} \ l_{\gamma}}&=\frac{1}{m_{\alpha}}\tr\left[V^{(k)} F^{\quad r_{\mu/\alpha} \   r_{\nu/\alpha}}_{k_{\beta} \ r_{\mu/\beta}  \ r_{\nu/\gamma} \ l_{\gamma}}\right]\\
																																																														 &=\frac{1}{\sqrt{m_{\mu}m_{\nu}}}\tr\left[V^{(k)} E^{r_{\mu/\beta}r_{\mu/\alpha}}_{k_{\beta}\quad 1_{\alpha}}V^{(k)}E^{r_{\nu/\alpha}r_{\nu/\gamma}}_{1_{\alpha}\quad l_{\gamma}}\right]. 
	\end{split}
	\ee
	Applying Fact~\ref{kpartr} and Lemma~\ref{L3} we reduce to
	\be
	\begin{split}
	 \left(V^{(k)}\right)^{\quad r_{\mu/\alpha} \   r_{\nu/\alpha}}_{k_{\beta} \ r_{\mu/\beta}  \ r_{\nu/\gamma} \ l_{\gamma}}&=\delta^{r_{\mu/\alpha}r_{\mu/\beta}}\frac{1}{\sqrt{m_{\mu}m_{\nu}}}\frac{m_{\mu}}{m_{\alpha}}\tr\left[E_{k_{\alpha}1_{\alpha}}^{\alpha} V^{(k)}E^{r_{\nu/\alpha}r_{\nu/\gamma}}_{1_{\alpha}\quad l_{\gamma}}\right]\\
	 &=\delta^{r_{\mu/\alpha}r_{\mu/\beta}}\frac{1}{m_{\alpha}}\sqrt{\frac{m_{\mu}}{m_{\nu}}}\tr\left[E_{k_{\alpha}1_{\alpha}}^{\alpha} E^{r_{\nu/\alpha}r_{\nu/\gamma}}_{1_{\alpha}\quad l_{\gamma}} \right], 
	 \end{split}
	\ee
	since only the operator $V^{(k)}$ acts non-trivially on last $k$ systems. Now, let us observe that the operator $E_{k_{\alpha}1_{\alpha}}^{\alpha}$ acts on first $n-2k$ systems, while the operator $E^{r_{\nu/\beta}}_{1_{\beta} \quad l_{\nu}}$ on $n-k$, so
	\be
	\begin{split}
	\left(V^{(k)}\right)^{\quad r_{\mu/\alpha} \   r_{\nu/\alpha}}_{k_{\beta} \ r_{\mu/\beta}  \ r_{\nu/\gamma} \ l_{\gamma}}&=\delta^{r_{\mu/\alpha}r_{\mu/\beta}}\frac{1}{m_{\alpha}}\sqrt{\frac{m_{\mu}}{m_{\nu}}}\tr\left[E_{k_{\alpha}1_{\alpha}}^{\alpha}\tr_{(k)}\left( E^{r_{\nu/\alpha}r_{\nu/\gamma}}_{1_{\alpha}\quad l_{\gamma}} \right) \right]\\
	&=\delta^{r_{\mu/\alpha}r_{\mu/\beta}}\delta^{r_{\nu/\alpha}r_{\nu/\gamma}}\frac{\sqrt{m_{\mu}m_{\nu}}}{m_{\alpha}m_{\gamma}}\tr\left[E_{k_{\alpha}1_{\alpha}}^{\alpha}E_{1_{\gamma}l_{\gamma}}^{\gamma} \right]\\
	&=\delta^{r_{\mu/\alpha}r_{\mu/\beta}}\delta^{r_{\nu/\alpha}r_{\nu/\gamma}}\frac{\sqrt{m_{\mu}m_{\nu}}}{m_{\alpha}^2}\tr E_{k_{\alpha}l_{\alpha}}^{\alpha}\\
	&=\delta^{r_{\mu/\alpha}r_{\mu/\beta}}\delta^{r_{\nu/\alpha}r_{\nu/\gamma}}\delta_{k_{\alpha}l_{\gamma}}\frac{\sqrt{m_{\mu}m_{\nu}}}{m_{\alpha}}.
	\end{split}
	\ee
	In the second equality we applied Lemma~\ref{L3}, while in fourth we used property from~\eqref{tr_prop}. Now we evaluate the matrix elements of $V_{\tau}$. Using expression~\eqref{rhs_action} we write
	\be
	\begin{split}
	\left( V_{\tau}\right)^{r_{\mu/\alpha}r_{\nu/\alpha}}_{i_{\mu}\quad j_{\nu}}& =\frac{1}{m_{\alpha}}\tr\left[V_{\tau} F^{r_{\mu/\alpha} r_{\nu/\alpha}}_{i_{\mu} \quad j_{\nu}}\right]\\
																																							& =\frac{1}{m_{\alpha}}\sum_{k_{\mu}}\phi_{k_{\mu}i_{\mu}}^{\mu}(\tau)\tr\left[F^{r_{\mu/\alpha} r_{\nu/\alpha}}_{k_{\mu} \quad j_{\nu}}\right]\\
	&=\frac{1}{\sqrt{m_{\mu}m_{\nu}}}\sum_{k_{\mu}}\phi_{k_{\mu}i_{\mu}}^{\mu}(\tau)\tr\left[E_{k_{\mu} \ 1_{\alpha}}^{\quad r_{\mu/\alpha}}V^{(k)}E^{r_{\nu/\alpha}}_{1_{\alpha}\quad  j_{\nu}}\right]\\
	&=\frac{1}{\sqrt{m_{\mu}m_{\nu}}}\sum_{k_{\mu}}\phi_{k_{\mu}i_{\mu}}^{\mu}(\tau)\tr\left[E_{k_{\mu} \ 1_{\alpha}}^{\quad r_{\mu/\alpha}}E^{r_{\nu/\alpha}}_{1_{\alpha}\quad  j_{\nu}}\right]\\
	&=\delta^{r_{\mu/\alpha}r_{\nu/\alpha}}\frac{1}{\sqrt{m_{\mu}m_{\nu}}}\sum_{k_{\mu}}\phi_{k_{\mu}i_{\mu}}^{\mu}(\tau)\tr\left(E^{\mu}_{k_{\mu}j_{\mu}} \right).
	\end{split}
   \ee
   Knowing that $\tr\left(E^{\mu}_{k_{\mu}j_{\mu}} \right)=\delta_{k_{\mu}j_{\mu}}m_{\mu}=\delta^{\mu\nu}\delta_{k_{\mu}j_{\nu}}m_{\mu}$ we simplify to
   \be
   \left( V_{\tau}\right)^{r_{\mu/\alpha}r_{\nu/\alpha}}_{i_{\mu}\quad j_{\nu}}=\delta^{r_{\mu/\alpha}r_{\nu/\alpha}}\delta_{k_{\mu}j_{\nu}}\sqrt{\frac{m_{\mu}}{m_{\nu}}}\sum_{k_{\mu}}\phi_{k_{\mu}i_{\mu}}^{\mu}(\tau).
   \ee   
    This finishes the proof.
\end{proof}
Having description of the basis elements in the ideal $\mathcal{M}$ and action properties we are ready to calculate matrix elements of the multi-port teleportation operator~\eqref{PBT1}.  
\begin{theorem}
	\label{kPBTmat}
	The matrix elements of the MPBT operator~\eqref{PBT1}, with number of ports $N$ and local dimension $d$, in operator basis from Theorem~\ref{tmbas} are of the form
	\be
	\label{kPBTmateq}
	(\rho)_{i_{\mu}\quad j_{\nu}}^{r_{\mu/\alpha}r_{\nu/\beta}}=\frac{k!\binom{N}{k}}{d^N}\frac{m_{\mu}}{m_{\alpha}}\frac{d_{\alpha}}{d_{\mu}}\delta^{r_{\mu/\alpha}r_{\nu/\beta}}\delta_{i_{\mu}j_{\nu}}.
	\ee
	The numbers $m_{\mu},m_{\alpha}$ and $d_{\mu},d_{\alpha}$ denote respective multiplicities and dimensions of the irrpes in the Schur-Weyl duality, labelled by $\alpha \vdash n-2k$ and $\mu \vdash n-k$, such that $\mu\in\alpha$.
\end{theorem}

\begin{proof}
The proof proceeds similarly as the proof of Lemma~\ref{Vel}, namely we compute
\be
\label{act0}
\begin{split}
(\rho)&_{i_{\mu}\quad j_{\nu}}^{r_{\mu/\alpha}r_{\nu/\beta}}\\
			& =\frac{1}{m_{\alpha}}\tr\left[\rho F^{r_{\mu/\alpha} r_{\nu/\alpha}}_{i_{\mu}\quad j_{\nu}} \right]\\
			&=\frac{1}{d^N}\frac{1}{\sqrt{m_{\mu}m_{\nu}}}\sum_{\tau \in \mathcal{S}_{n,k}}\tr\left[V^{(k)}V_{\tau} E_{i_{\mu} \ 1_{\alpha}}^{\quad r_{\mu/\alpha}}V^{(k)}E^{r_{\nu/\alpha}}_{1_{\alpha} \quad j_{\nu}}V_{\tau^{-1}}\right], 
																													 \end{split}
\ee
where sum runs over all permutations $\tau$ from the coset $\mathcal{S}_{n,k}\equiv \frac{S(n-k)}{S(n-2k)}$.  Substituting~\eqref{actionEE} to~\eqref{act0} we have
\be
\begin{split}
\label{rho-sro}
(\rho)_{i_{\mu}\quad j_{\nu}}^{r_{\mu/\alpha}r_{\nu/\beta}}=\frac{1}{d^N}\frac{1}{\sqrt{m_{\mu}m_{\nu}}}&\sum_{\tau \in \mathcal{S}_{n,k}}\sum_{l_{\mu}}\sum_{k_{\nu}}\phi^{\mu}_{l_{\mu}i_{\mu}}(\tau)\phi_{j_{\nu}k_{\nu}}^{\nu}(\tau^{-1})\\
			& \times \tr\left[V^{(k)}E_{l_{\mu} \ 1_{\alpha}}^{\quad r_{\mu/\alpha}}V^{(k)}E_{1_{\alpha}\quad k_{\nu}}^{r_{\nu/\alpha}}\right].
\end{split}
\ee
Using Fact~\ref{kpartr} we write the following chain of equalities:
\be
\begin{split}
	\tr&\left[V^{(k)} E_{l_{\mu} \ 1_{\alpha}}^{\quad r_{\mu/\alpha}}V^{(k)}E_{1_{\alpha}\quad k_{\nu}}^{r_{\nu/\alpha}}\right]\\
&=\tr\left[\tr_{(k)}\left(E_{l_{\mu} \ 1_{\alpha}}^{\quad r_{\mu/\alpha}} \right) V^{(k)}E_{1_{\alpha}\quad k_{\nu}}^{r_{\nu/\alpha}} \right]\\
									&=\tr\left[\tr_{(k)}\left(E_{l_{\mu} \ 1_{\alpha}}^{\quad r_{\mu/\alpha}} \right) E_{1_{\alpha}\quad k_{\nu}}^{r_{\nu/\alpha}} \right]\\
&=\tr\left[\tr_{(k)}\left(E_{l_{\mu} \ 1_{\alpha}}^{\quad r_{\mu/\alpha}} \right) \tr_{(k)}\left(E_{1_{\alpha}\quad k_{\nu}}^{r_{\nu/\alpha}}\right)  \right],
\end{split}
\ee
where $\tr_{(k)}=\tr_{n-2k+1,\ldots,n-k}$. Expanding rest of the indices in PRIR notation, i.e. $l_{\mu}=(s_{\mu/\beta},p_{\beta}), \ k_{\nu}=(s_{\nu/\beta'},q_{\beta'})$ and applying Lemma~\ref{L3} we have
\be
\begin{split}
&\tr_{(k)}\left(E_{p_{\beta} \quad1_{\alpha}}^{s_{\mu/\beta}  r_{\mu/\alpha}} \right)=\delta^{s_{\mu/\beta}r_{\mu/\alpha}}\frac{m_{\mu}}{m_{\alpha}}E^{\alpha}_{p_{\alpha}1_{\alpha}}\\
&\tr_{(k)}\left(E_{1_{\alpha} \quad q_{\beta'}}^{r_{\nu/\alpha}  s_{\nu/\beta'}}\right)=\delta^{r_{\nu/\alpha}s_{\nu/\beta'}}\frac{m_{\nu}}{m_{\alpha}}E^{\alpha}_{1_{\alpha}q_{\alpha}}.
\end{split}
\ee
Now, we substitute the above into~\eqref{rho-sro} writing as follows
\begin{IEEEeqnarray}{lll}
	(\rho)&_{i_{\mu}\quad j_{\nu}}^{r_{\mu/\alpha}r_{\nu/\beta}} & \\
				&=\frac{1}{d^N}\frac{\sqrt{m_{\mu}m_{\nu}}}{m^2_{\alpha}}\sum_{\tau \in \mathcal{S}_{n,k}}\sum_{r_{\mu/\alpha},p_{\alpha}}\sum_{r_{\nu/\alpha},q_{\alpha}} & \left( \phi^{\nu}\right) _{j_{\nu} \ q_{\alpha}}^{ \quad r_{\nu/\alpha}}(\tau^{-1})\left( \phi^{\mu}\right)_{p_{\alpha} \  \ i_{\mu}}^{r_{\mu/\alpha}}(\tau)\nonumber\\
				& & \times \tr\left(E^{\alpha}_{p_{\alpha}1_{\alpha}} E^{\alpha}_{1_{\alpha}q_{\alpha}}\right)\nonumber\\
				&=\frac{1}{d^N}\frac{\sqrt{m_{\mu}m_{\nu}}}{m_{\alpha}}\sum_{\tau \in \mathcal{S}_{n,k}}\sum_{r_{\mu/\alpha},p_{\alpha}}\sum_{r_{\nu/\alpha},q_{\alpha}}& \delta_{p_{\alpha}q_{\alpha}}\left( \phi^{\nu}\right) _{j_{\nu} \ q_{\alpha}}^{ \quad r_{\nu/\alpha}}(\tau^{-1}) \nonumber\\
				& & \times\left( \phi^{\mu}\right)_{p_{\alpha} \  \ i_{\mu}}^{r_{\mu/\alpha}}(\tau) \nonumber\\
				&=\frac{1}{d^N}\frac{\sqrt{m_{\mu}m_{\nu}}}{m_{\alpha}}\sum_{\tau \in \mathcal{S}_{n,k}}\sum_{r_{\mu/\alpha},r_{\nu/\alpha}}\sum_{q_{\alpha}}& \left( \phi^{\nu}\right) _{j_{\nu} \ q_{\alpha}}^{ \quad r_{\nu/\alpha}}(\tau^{-1})\left( \phi^{\mu}\right)_{q_{\alpha} \  \ i_{\mu}}^{r_{\mu/\alpha}}(\tau).\nonumber\\
\end{IEEEeqnarray}
In the above we use orthonormality relation, together wit the trace property~\eqref{tr_prop}, so $\tr\left(E^{\alpha}_{p_{\alpha}1_{\alpha}} E^{\alpha}_{1_{\alpha}q_{\alpha}}\right)=\tr E^{\alpha}_{p_{\alpha}q_{\alpha}}=m_{\alpha}\delta_{p_{\alpha}q_{\alpha}}$. Finally applying summation rule from Proposition~\ref{summation0} we arrive at
\be
(\rho)_{i_{\mu}\quad j_{\nu}}^{r_{\mu/\alpha}r_{\nu/\beta}}=\frac{1}{d^N}|\mathcal{S}_{n,k}|\frac{m_{\mu}}{m_{\alpha}}\frac{d_{\alpha}}{d_{\mu}}=\frac{k!\binom{N}{k}}{d^N}\frac{m_{\mu}}{m_{\alpha}}\frac{d_{\alpha}}{d_{\mu}}\delta^{r_{\mu/\alpha}r_{\nu/\beta}}\delta_{i_{\mu}j_{\nu}}.
\ee
This finishes the proof.
\end{proof}

Let us check the consequences of Theorem~\ref{kPBTmat}. Expression~\eqref{kPBTmateq} tells us that multi-port teleportation operator $\rho$ is diagonal in the operator basis given in Theorem~\ref{tmbas}. It means $\rho$ can expressed as
\be
\begin{split}
\label{rhodec0}
\rho&=\frac{k!\binom{N}{k}}{d^n}\sum_{\alpha}\sum_{\mu\in\alpha}\sum_{r_{\mu/\alpha}}\sum_{k_{\mu}}\frac{m_{\mu}}{m_{\alpha}}\frac{d_{\alpha}}{d_{\mu}}F^{r_{\mu/\alpha}r_{\mu/\alpha}}_{k_{\mu} \quad k_{\mu}}\\
		&=\sum_{\alpha}\sum_{\mu\in\alpha}\sum_{r_{\mu/\alpha}}\sum_{k_{\mu}}\lambda_{\mu}(\alpha)F^{r_{\mu/\alpha}r_{\mu/\alpha}}_{k_{\mu} \quad k_{\mu}}
\end{split}
\ee
where we introduced the quantity
\be
\label{rhodec0a}
\lambda_{\mu}(\alpha)\equiv \frac{k!\binom{N}{k}}{d^N}\frac{m_{\mu}}{m_{\alpha}}\frac{d_{\alpha}}{d_{\mu}}.
\ee
Now we can formulate the following
\begin{definition}
\label{efy}
Having basis elements from~\eqref{tmbas2} of Theorem~\ref{tmbas}, we define the following operators
\be
\forall \alpha \ \forall \mu\in\alpha \quad F_{\mu}(\alpha)\equiv \sum_{r_{\mu/\alpha}}\sum_{k_{\mu}}F^{r_{\mu/\alpha}r_{\mu/\alpha}}_{k_{\mu} \quad k_{\mu}}.
\ee
\end{definition}
Having the above definition we prove:
\begin{lemma}
\label{efy2}
Operators $F_{\mu}(\alpha)$ for $\alpha\vdash N-k$ and $\mu\in\alpha$ are projectors and span identity $\mathbf{1}_{\mathcal{M}}$ on the ideal $\mathcal{M}$.
\end{lemma}

\begin{proof}
First let us check that operators $F_{\mu}(\alpha)$ given through Definition~\ref{efy} are indeed orthonormal projectors. Indeed using~\eqref{orto} we have
\be
\begin{split}
F_{\mu}(\alpha)F_{\nu}(\beta)&=\sum_{r_{\mu/\alpha}}\sum_{k_{\mu}}\sum_{r_{\nu/\beta}}\sum_{l_{\nu}}F^{r_{\mu/\alpha}r_{\mu/\alpha}}_{k_{\mu} \quad k_{\mu}}F^{r_{\nu/\beta}r_{\nu/\beta}}_{l_{\nu} \quad l_{\nu}}\\
														 &=\sum_{r_{\mu/\alpha}}\sum_{k_{\mu}}\sum_{r_{\nu/\beta}}\sum_{l_{\nu}}\delta^{r_{\mu/\alpha}r_{\nu/\beta}}\delta_{k_{\mu}l_{\nu}}F^{r_{\mu/\alpha}r_{\nu/\beta}}_{k_{\mu} \quad l_{\nu}}\\
&=\delta^{\mu\nu}\delta^{\alpha\beta}\sum_{r_{\mu/\alpha}}\sum_{k_{\mu}}\sum_{r_{\nu/\beta}}\sum_{l_{\nu}}\delta^{r_{\mu/\alpha}r_{\nu/\beta}}\delta_{k_{\mu}l_{\nu}}F^{r_{\mu/\alpha}r_{\nu/\beta}}_{k_{\mu} \quad l_{\nu}}\\
&=\delta^{\mu\nu}\delta^{\alpha\beta}\sum_{r_{\mu/\alpha}}\sum_{k_{\mu}}F^{r_{\mu/\alpha}r_{\mu/\alpha}}_{k_{\mu} \quad k_{\mu}}\\
&=\delta^{\mu\nu}\delta^{\alpha\beta}F_{\mu}(\alpha),
\end{split}
\ee
since for fixed $\mu,\nu$ and $\alpha,\beta$ we use the property $\delta^{r_{\mu/\alpha}r_{\nu/\beta}}\equiv \delta^{\mu\nu}\delta^{\alpha\beta}\delta^{r_{\mu/\alpha}r_{\nu/\beta}}$, see Notation~\ref{not0}.

To prove $\sum_{\alpha}\sum_{\mu\in\alpha}F_{\mu}(\alpha)=\mathbf{1}_{\mathcal{M}}$ we must show that $\forall x\in\mathcal{M}$ we have $x\sum_{\alpha}\sum_{\mu\in\alpha}F_{\mu}(\alpha)=\sum_{\alpha}\sum_{\mu\in\alpha}F_{\mu}(\alpha)x=x$. Expanding $x$ in the operator basis from Theorem~\ref{tmbas}
\begin{multline}
x=\sum_{\alpha',\beta'}\sum_{\mu'\in\alpha'}\sum_{\nu'\in\beta'}\sum_{i_{\mu'} \ j_{\nu'}}x_{i_{\mu'}\quad j_{\nu'}}^{r_{\mu'/\alpha'} r_{\nu'/\beta'}}F_{i_{\mu'}\quad j_{\nu'}}^{r_{\mu'/\alpha'} r_{\nu'/\beta'}},\\ x_{i_{\mu'}\quad j_{\nu'}}^{r_{\mu'/\alpha'} r_{\nu'/\beta'}}\in \mathbb{C},
\end{multline}
and using expression~\eqref{orto} we get the statement.
\end{proof}

Finally thanks to Lemma~\ref{efy2} and decomposition~\eqref{rhodec0}, together with~\eqref{rhodec0a} we formulate spectral theorem for the multi-port teleportation operator (the multiplicities given below come from Lemma~\ref{A3}):
\begin{theorem}
\label{eig_dec_rho}
The MPBT operator given through~\eqref{PBT1} has the following spectral decomposition
\be
\rho=\sum_{\alpha}\sum_{\mu\in\alpha}\lambda_{\mu}(\alpha)F_{\mu}(\alpha),
\ee
where eigenprojectors $F_{\mu}(\alpha)$ are given in Definition~\ref{efy} with corresponding eigenvalues $\lambda_{\mu}(\alpha)$ from~\eqref{rhodec0a} with multiplicities $m_{\mu/\alpha}m_{\alpha}d_{\mu}$.
\end{theorem}
\noindent
Checking that indeed we have $\rho F_{\mu}(\alpha)=\lambda_{\mu}(\alpha)F_{\mu}(\alpha)$, follows directly from orthonormality property of operators $F_{\mu}(\alpha)$ proven in Lemma~\ref{efy2}. 

At the end of this section we prove two additionally lemmas on projectors $F_{\mu}(\alpha)$ given in Definition~\ref{efy}. Defining symbol the $\tr_{(2k)}\equiv \tr_{n-2k+1,\ldots,n}$ which is a partial trace operation with respect to last $2k$ systems we have the following
\begin{lemma}
\label{A1}
For a partially transposed permutation operator $V^{(k)}$ from~\eqref{parV} and operators $F_{\mu}(\alpha)$ given through Definition~\ref{efy} the following holds:
\be
\forall \alpha\vdash n-2k \quad \forall \mu\in \alpha \quad \tr_{(2k)}\left[V^{(k)} F_{\mu}(\alpha)\right]=m_{\mu/\alpha}\frac{m_{\mu}}{m_{\alpha}}P_{\alpha},
\ee
where the numbers $m_{\mu},m_{\alpha}$ denote respective multiplicities in the Schur-Weyl duality, while $P_{\alpha}$ is a Young projector on $n-2k$ particles.
\end{lemma}

\begin{proof}
Using definition of the operator $F_{\mu}(\alpha)$ and expression~\eqref{tmbas2} we write
\begin{multline}
\sum_{r_{\mu/\alpha}}\sum_{k_{\mu}}\tr_{(2k)}\left[V^{(k)} F^{r_{\mu/\alpha}r_{\mu/\alpha}}_{k_{\mu} \quad k_{\mu}}\right]\\
=\sum_{r_{\mu/\alpha}}\sum_{k_{\mu}}\frac{m_{\alpha}}{m_{\mu}}\tr_{(2k)}\left[V^{(k)}E_{i_{\mu} \ 1_{\alpha}}^{\quad r_{\mu/\alpha}}V^{(k)}E^{r_{\mu/\alpha}}_{1_{\alpha} \quad  k_{\mu}}\right].
\end{multline}
Using Fact~\ref{kpartr}, Lemma~\ref{L3} and $i_{\mu}=(s_{\mu/\beta},i_{\beta})$ to operator $E_{k_{\mu} \quad 1_{\alpha}}^{\quad r_{\mu/\alpha}}=E_{k_{\beta} \  \ 1_{\alpha}}^{r_{\mu/\beta}   \ r_{\mu/\alpha}}$ we simplify the above equation to
\be
\begin{split}
	\sum_{r_{\mu/\alpha}}\sum_{k_{\alpha}}&\tr_{(2k)}\left[E^{\alpha}_{k_{\alpha}1_{\alpha}}V^{(k)}E^{r_{\mu/\alpha} r_{\mu/\alpha}}_{1_{\alpha} \quad  k_{\alpha}} \right]\\
	                                      &=\sum_{r_{\mu/\alpha}}\sum_{k_{\alpha}}\tr_{(k)}\left[E^{\alpha}_{k_{\alpha}1_{\alpha}}E^{r_{\mu/\alpha} r_{\mu/\alpha}}_{1_{\alpha} \quad  k_{\alpha}} \right]\\
																				&=\sum_{r_{\mu/\alpha}} \sum_{k_{\alpha}}\tr_{(k)}\left[E^{r_{\mu/\alpha} r_{\mu/\alpha}}_{k_{\alpha} \quad  k_{\alpha}} \right]\\
&=m_{\mu/\alpha}\frac{m_{\mu}}{m_{\alpha}}\sum_{k_{\alpha}}E^{\alpha}_{k_{\alpha}k_{\alpha}}\\
&=m_{\mu/\alpha}\frac{m_{\mu}}{m_{\alpha}}P_{\alpha},
\end{split}
\ee
where in the last equality we used the definition of projectors $P_{\alpha}$ given in~\eqref{def_P}. 
\end{proof}
Further, below the proof of Lemma~\ref{simple} we discuss alternative proof method of the above lemma.

\begin{lemma}
	\label{A2}
For operators $F_{\mu}(\alpha)$ given through Definition~\ref{efy} the following holds:
\be
\label{A2eq}
\forall \alpha\vdash n-2k \quad \forall \mu\in \alpha \quad \tr_{(k)}\left( F_{\mu}(\alpha)\right)=m_{\mu/\alpha}\frac{m_{\alpha}}{m_{\mu}}P_{\mu},
\ee
where the numbers $m_{\mu},m_{\alpha}$ denote respective multiplicities of the irrpes in the Schur-Weyl duality, $m_{\mu/\alpha}$ denotes number of paths on reduced Young's lattice in which diagram $\mu$ can be obtained from diagram $\alpha$, while $P_{\mu}$ is a Young projector on $n-k$ particles.	
\end{lemma}

\begin{proof}
The proof is based on the straightforward calculations and observations made in the proof of Lemma~\ref{A1}. Using Definition~\ref{efy} we have
\be
\begin{split}
\tr_{(k)}\left( F_{\mu}(\alpha)\right)&=\frac{m_{\alpha}}{m_{\mu}}\sum_{r_{\mu/\alpha}}\sum_{k_{\mu}}\tr_{(k)}\left(E_{k_{\mu} \ 1_{\alpha}}^{\quad r_{\mu/\alpha}}V^{(k)}E^{r_{\mu/\alpha}}_{1_{\alpha} \quad  k_{\mu}} \right)\\
																			&=\frac{m_{\alpha}}{m_{\mu}}\sum_{r_{\mu/\alpha}}\sum_{k_{\mu}} E_{k_{\mu} \ 1_{\alpha}}^{\quad r_{\mu/\alpha}}E^{r_{\mu/\alpha}}_{1_{\alpha} \quad  k_{\mu}}\\
&=m_{\mu/\alpha}\frac{m_{\alpha}}{m_{\mu}}\sum_{k_{\mu}}E^{\mu}_{k_{\mu}k_{\mu}}=m_{\mu/\alpha}\frac{m_{\alpha}}{m_{\mu}}P_{\mu}
\end{split}
\ee 
where in the last equality we used the definition of projectors $P_{\mu}$ given in~\eqref{def_P}.
\end{proof}

\begin{lemma}
\label{A3}
For operators $F_{\mu}(\alpha)$ given through Definition~\ref{efy} the following holds:
\be
\forall \alpha\vdash n-2k \quad \forall \mu\in \alpha \quad \tr\left( F_{\mu}(\alpha)\right)=m_{\mu/\alpha}m_{\alpha}d_{\mu},
\ee
where the numbers $m_{\mu},m_{\alpha}$ denote respective multiplicities of irreps in the Schur-Weyl duality, $d_{\mu}$ stands for the dimension of the irrep $\mu$, $m_{\mu/\alpha}$ denotes number of paths on reduced Young's lattice in which diagram $\mu$ can be obtained from diagram $\alpha$.
\end{lemma}

\begin{proof}
To compute the trace from $F_{\mu}(\alpha)$ is enough to compute the trace from the right-hand side of~\eqref{A2eq} of Lemma~\ref{A2}, knowing that $\tr P_{\mu}=m_{\mu}d_{\mu}$.
\end{proof}

\begin{lemma}
\label{simple}
For operators $F_{\mu}(\alpha)$ given through Definition~\ref{efy} and operator $V^{(k)}$ defined in~\eqref{parV}, the following holds:
\be
\label{cos}
V^{(k)}F_{\mu}(\alpha)=V^{(k)}P_{\alpha}P_{\mu}.
\ee
\end{lemma}

\begin{proof}
First let us write explicitly the left-hand side of~\eqref{cos} using Definition~\ref{efy} and Lemma~\ref{L3}:
\be
\label{p1}
\begin{split}
	V^{(k)}F_{\mu}(\alpha)&=V^{(k)}\sum_{r_{\mu/\alpha}}\sum_{k_{\mu}}F^{r_{\mu/\alpha}r_{\mu/\alpha}}_{k_{\mu} \quad k_{\mu}}\\
	                      &=\frac{m_{\alpha}}{m_{\mu}}\sum_{r_{\mu/\alpha}}\sum_{r_{\mu/\beta}}\sum_{i_{\beta}}V^{(k)}E^{r_{\mu/\beta}r_{\mu/\alpha}}_{i_{\beta}\quad 1_{\alpha}}V^{(k)}E^{r_{\mu/\alpha}r_{\mu/\beta}}_{1_{\alpha}\quad i_{\beta}}\\
												&=V^{(k)}\sum_{r_{\mu/\alpha}}\sum_{i_{\alpha}}E^{\alpha}_{i_{\alpha}1_{\alpha}}E^{r_{\mu/\alpha}r_{\mu/\alpha}}_{1_{\alpha}\quad i_{\alpha}}\\
												&=V^{(k)}\sum_{r_{\mu/\alpha}}\sum_{i_{\alpha}}E^{r_{\mu/\alpha}r_{\mu/\alpha}}_{i_{\alpha}\quad i_{\alpha}}.
\end{split}
\ee
Now, writing composition $P_{\alpha}P_{\mu}$ in PRIR basis we get:
\be
\label{p2}
\begin{split}
	V^{(k)}P_{\alpha}P_{\mu}&=V^{(k)}\sum_{i_{\alpha}}E^{\alpha}_{i_{\alpha}i_{\alpha}}\sum_{r_{\mu/\beta}}\sum_{j_{\beta}}E^{r_{\mu/\beta}r_{\mu/\beta}}_{j_{\beta}\quad j_{\beta}}\\
	                        &=V^{(k)}\sum_{r_{\mu/\alpha}}\sum_{i_{\alpha}}E^{r_{\mu/\alpha}r_{\mu/\alpha}}_{i_{\alpha}\quad i_{\alpha}}
\end{split}
\ee
since $E^{\alpha}_{i_{\alpha}i_{\alpha}}E^{r_{\mu/\beta}r_{\mu/\beta}}_{j_{\beta}\quad j_{\beta}}=\delta^{\alpha\beta}\delta_{i_{\alpha}j_{\beta}}E^{r_{\mu/\alpha}r_{\mu/\alpha}}_{i_{\alpha}\quad i_{\alpha}}$. Now observing  that right-hand sides of~\eqref{p1} and~\eqref{p2} coincide we finish the proof.
\end{proof}
One can observe that having~\eqref{cos} we can prove the statement of Lemma~\ref{A1} applying directly Corollary~\ref{corL3} to projector $P_{\mu}$. Indeed we have
\be
\begin{split}
\tr_{(2k)}\left( V^{(k)}F_{\mu}(\alpha)\right) & =\tr_{(2k)}\left(V^{(k)}P_{\alpha}P_{\mu}\right)\\
&=\tr_{(k)}\left(P_{\alpha}P_{\mu} \right)=m_{\mu/\alpha}\frac{m_{\mu}}{m_{\alpha}}P_{\alpha},
\end{split}
\ee
where $\tr_{(2k)}\equiv \tr_{n-2k+1,\ldots,n}$ and $\tr_{(k)}=\tr_{n-2k+1,\ldots,n-k}$.

\section{Entanglement fidelity in Deterministic version of the protocol}
\label{detkPBT}
 Having description of the deterministic version of MPBT from Section~\ref{interest} and  mathematical tools developed in Section~\ref{comm_structure}, especially the spectral decomposition of the operator $\rho$, given in Theorem~\ref{eig_dec_rho}, we can formulate the following:
\begin{theorem}
	\label{Fthm}
	The entanglement fidelity in the deterministic multi-port teleportation with $N$ ports and local dimension $d$ is given as
	\be
	\label{Feq1}
	F=\frac{1}{d^{N+2k}}\sum_{\alpha \vdash N-k}\left(\sum_{\mu\in\alpha}m_{\mu/\alpha} \sqrt{m_{\mu}d_{\mu}}\right)^2,
	\ee
	where $m_{\mu},d_{\mu}$ denote multiplicity and dimension of irreducible representations of $S(N)$ respectively, and $m_{\mu/\alpha}$ denotes number of paths on reduced Young's lattice in which diagram $\mu$ can be obtained from diagram $\alpha$ by adding $k$ boxes.
\end{theorem}

\begin{proof}
In the first step of the proof we apply the covariance property~\eqref{rel1} and~\eqref{rel2} to equation~\eqref{ent_fid} describing the entanglement fidelity and obtain the following expression:
\be
\begin{split}
\label{ef}
F=\frac{1}{d^{2k}}\sum_{\mathbf{i}\in\mathcal{I}}\tr\left(\Pi_{\mathbf{i}}^{A\widetilde{B}}\sigma_{\mathbf{i}}^{A\widetilde{B}} \right)& =\frac{|\mathcal{S}_{n,k}|}{d^{2k}}\tr\left(\Pi_{\mathbf{i}_0}^{A\widetilde{B}}\sigma_{\mathbf{i}_0}^{A\widetilde{B}} \right)\\
                                                                                                                                       &=\frac{k!\binom{N}{k}}{d^{2k}}\tr\left(\frac{1}{\sqrt{\rho}}\sigma_{\mathbf{i}_0}^{A\widetilde{B}}\frac{1}{\sqrt{\rho}}\sigma_{\mathbf{i}_0}^{A\widetilde{B}} \right), 
\end{split}
\ee
where $\sigma_{\mathbf{i}_0}^{A\widetilde{B}}$ is defined in~\eqref{signal2}. In the second equality we used the covariance property of signals $\sigma_{\mathbf{i}}^{A\widetilde{B}}$ and invariance of $\rho$ with respect to the coset $\mathcal{S}_{n,k}$.
Using spectral decomposition of the operator $\rho$ presented in Theorem~\ref{eig_dec_rho} we expand equation~\eqref{ef} to:
\begin{IEEEeqnarray}{lll}
	F&=\frac{k!}{d^{2k}}\binom{N}{k}\tr\left(\Pi_{\mathbf{i}_0}^{A\widetilde{B}}\sigma_{\mathbf{i}_0}^{A\widetilde{B}} \right) & \nonumber\\
	 &=\frac{k!}{d^{2N}}\binom{N}{k}\sum_{\alpha,\beta \vdash N-k}\sum_{\mu\in\alpha}\sum_{\nu\in\beta} & \frac{1}{\sqrt{\lambda_{\mu}(\alpha)}}\frac{1}{\sqrt{\lambda_{\nu}(\beta)}}\nonumber\\
	 & & \times\tr\left(F_{\mu}(\alpha)V^{(k)}F_{\nu}(\beta)V^{(k)} \right).\nonumber\\
\label{fak}
\end{IEEEeqnarray}
Now applying Lemma~\ref{simple} we can rid of the operators $F_{\mu}(\alpha)$ 
\begin{IEEEeqnarray}{lll}
F&=\frac{k!}{d^{2N}}\binom{N}{k}\sum_{\alpha,\beta \vdash N-k}\sum_{\mu\in\alpha}\sum_{\nu\in\beta}&\frac{1}{\sqrt{\lambda_{\mu}(\alpha)}}\frac{1}{\sqrt{\lambda_{\nu}(\beta)}}\nonumber\\ 
 &&\tr\left(F_{\mu}(\alpha)V^{(k)}F_{\nu}(\beta)V^{(k)} \right)\nonumber\\
 &=\frac{k!}{d^{2N}}\binom{N}{k}\sum_{\alpha,\beta \vdash N-k}\sum_{\mu\in\alpha}\sum_{\nu\in\beta}&\frac{1}{\sqrt{\lambda_{\mu}(\alpha)}}\frac{1}{\sqrt{\lambda_{\nu}(\beta)}}\nonumber\\
 &&\tr\left(P_{\mu}P_{\alpha}V^{(k)}P_{\nu}P_{\beta}V^{(k)} \right).\nonumber\\
\label{fak2}
\end{IEEEeqnarray}
Observing $\left[ P_{\beta},V^{(k)}\right] =0$, we can apply Fact~\ref{kpartr} together with Corollary~\ref{corL3} to $V^{(k)}P_{\nu}V^{(k)}$, getting
\begin{IEEEeqnarray}{lll}
	F&=\frac{k!}{d^{2N}}\binom{N}{k}&\sum_{\alpha,\beta \vdash N-k}\sum_{\mu\in\alpha}\sum_{\nu\in\beta}\frac{1}{\sqrt{\lambda_{\mu}(\alpha)}}\frac{1}{\sqrt{\lambda_{\nu}(\beta)}}\nonumber\\
	 &&\times \sum_{\beta'\in\nu}m_{\nu/\beta'}\frac{m_{\nu}}{m_{\beta'}}\tr\left(P_{\mu}P_{\alpha}P_{\beta}P_{\beta'} V^{(k)}\right) \nonumber \\
	 &=\frac{k!}{d^{2N}}\binom{N}{k}&\sum_{\alpha \vdash N-k}\sum_{\mu,\nu\in\alpha}\frac{1}{\sqrt{\lambda_{\mu}(\alpha)}}\frac{1}{\sqrt{\lambda_{\nu}(\alpha)}}m_{\nu/\alpha}\frac{m_{\nu}}{m_{\alpha}}\nonumber\\ 
	 & & \times \tr\left(P_{\mu}P_{\alpha}\tr_{(k)}V^{(k)}\right)\nonumber \\
	 &=\frac{k!}{d^{2N}}\binom{N}{k}&\sum_{\alpha \vdash N-k}\sum_{\mu,\nu\in\alpha}\frac{1}{\sqrt{\lambda_{\mu}(\alpha)}}\frac{1}{\sqrt{\lambda_{\nu}(\alpha)}}m_{\nu/\alpha}\frac{m_{\nu}}{m_{\alpha}}\nonumber\\
& &\times \tr\left(P_{\mu}P_{\alpha}\right).\nonumber\\
\end{IEEEeqnarray}
Again applying Corollary~\ref{corL3}, this time to projector $P_{\mu}$, together with $\tr P_{\alpha}=m_{\alpha}d_{\alpha}$, we have
\begin{multline}
F=\\
\frac{k!}{d^{2N}}\binom{N}{k}\sum_{\alpha \vdash N-k}\sum_{\mu,\nu\in\alpha}\frac{1}{\sqrt{\lambda_{\mu}(\alpha)}}\frac{1}{\sqrt{\lambda_{\nu}(\alpha)}}m_{\nu/\alpha}m_{\mu/\alpha}m_{\mu}m_{\nu}\frac{d_{\alpha}}{m_{\alpha}}.
\end{multline}
Using explicit expression for eigenvalues $\lambda_{\mu}(\alpha),\lambda_{\nu}(\alpha)$ given in~\eqref{rhodec0a} we have
\begin{IEEEeqnarray}{lll}
	F&=\frac{k!}{d^{2N+2k}}\binom{N}{k}\frac{d^{N}}{k!\binom{N}{k}}\sum_{\alpha \vdash N-k}&\sum_{\mu,\nu\in\alpha}m_{\mu/\alpha}m_{\nu/\alpha}\nonumber \\
	 && \times\sqrt{\frac{m_{\alpha}d_{\mu}}{m_{\mu}d_{\alpha}}}\sqrt{\frac{m_{\alpha}d_{\nu}}{m_{\nu}d_{\alpha}}}\frac{m_{\mu}m_{\nu}}{m_{\alpha}}d_{\alpha}\nonumber \\
	 &=\frac{1}{d^{N+2k}}\sum_{\alpha \vdash N-k}\sum_{\mu,\nu\in\alpha}m_{\mu/\alpha}&\sqrt{m_{\mu}d_{\mu}}m_{\nu/\alpha}\sqrt{m_{\nu}d_{\nu}}\nonumber\\
	 &=\frac{1}{d^{N+2k}}\sum_{\alpha \vdash N-k}\bigg(\sum_{\mu\in\alpha}m_{\mu/\alpha} &\sqrt{m_{\mu}d_{\mu}}\bigg)^2.\nonumber\\
\end{IEEEeqnarray}
This finishes the proof.
\end{proof}
An alternative proof of Theorem~\ref{Fthm} is presented in Appendix~\ref{AppA}. One can see that by setting $k=1$ to~\eqref{Feq1} we reproduce known expression for entanglement fidelity in ordinary port-based teleportation~\cite{Studzinski2017}. Indeed, in this case always $m_{\mu/\alpha}=1$, for any $\mu\in\alpha$, since we can move only by one layer on reduced Young's lattice.
The expression from~\eqref{Feq1} is plotted in Figure~\ref{fig:test2} for different number of ports $N$ as well local dimension $d$ and number of teleported states $k$. We  see that our deterministic scheme performs significantly better than standard PBT protocol, even in the optimal scheme, with respective dimension of the port. 

\section{Probability of success in Probabilistic version of the protocol}
\label{probkPBT}
 Having description of the probabilistic version of MPBT scheme from Section~\ref{interest} we are in position to solve SDP programs and evaluate optimal probability of success $p$ when the parties share maximally entangled states. Namely, we have the following:
\begin{theorem}
\label{thm_p}
The average probability of success in the probabilistic multi-port teleportation with $N$ ports and local dimension $d$ is given as
\be
\label{exact}
p=\frac{k!\binom{N}{k}}{d^{2N}}\sum_{\alpha \vdash N-k}\mathop{\operatorname{min}}\limits_{\mu\in\alpha}\frac{m_{\alpha}d_{\alpha}}{\lambda_{\mu}(\alpha)},
\ee
with optimal measurements of the form
\be
\label{ex_measurements}
\forall \ \mathbf{i}\in\mathcal{I}\qquad \Pi_{\mathbf{i}}^{AC}=\frac{k!\binom{N}{k}}{d^{2N}}P^+_{A_{\mathbf{i}}C}\ot \sum_{\alpha \vdash N-k}P_{\alpha}\mathop{\operatorname{min}}\limits_{\mu\in\alpha}\frac{1}{\lambda_{\mu}(\alpha)}.
\ee
Numbers $\lambda_{\mu}(\alpha)$ are eigenvalues of $\rho$ and are given in~\eqref{rhodec0a} and $m_{\alpha}, d_{\alpha}$ denote multiplicity and dimension of the irrep labelled by $\alpha$.
\end{theorem}
\begin{proof}
The solution of optimisation tasks, so proof of the above theorem, is based solely on methods and tools delivered in Section~\ref{preliminary} and Section~\ref{comm_structure}. We start from solving the primal problem. Due to symmetry in our scheme we assume that $\forall \mathbf{i}\in\mathcal{I} \quad \Theta_{\overline{A}_{\mathbf{i}}}=\sum_{\alpha \vdash N-k}x_{\alpha}P_{\alpha}$ with $x_{\alpha}\geq 0$ to satisfy constraint (1) from~\eqref{con1}. Operators $P_{\alpha}$ are Young projectors acting on subsystems defined by the symbol $\overline{A}_{\mathbf{i}}$.  To satisfy constraint (2) from~\eqref{con1} we write for every irreducible block $\alpha$:
\be
\begin{split}
\label{ineq}
\sum_{\mathbf{i}\in \mathcal{I}}P^+_{A_{\mathbf{i}}C}\ot \Theta_{\overline{A}_{\mathbf{i}}}(\alpha)&=\frac{x_{\alpha}}{d^k}\sum_{\tau\in\mathcal{S}_{n,k}}V_{\tau^{-1}}V^{(k)}\ot P_{\alpha}V_{\tau}\\
																																																	 &=d^{N-k}x_{\alpha}\rho(\alpha)\leq P_{\alpha}.
\end{split}
\ee
In the above expression we use fact that for operator $\rho$ from~\eqref{PBT1} and projection $P_{\alpha}$ we have $\rho(\alpha)=P_{\alpha}\rho P_{\alpha}$. Now, to satisfy inequality~\ref{ineq} it is enough to require:
\be
\label{border}
\forall  \alpha \quad x_{\alpha}\leq d^{k-N}\mathop{\operatorname{min}}\limits_{\mu\in\alpha}\frac{1}{\lambda_{\mu}(\alpha)},
\ee
where numbers $\lambda_{\mu}(\alpha)$ are eigenvalues of $\rho$ and are given in~\eqref{rhodec0a}.
Using assumption of covariance of measurements $\forall \ \tau\in \mathcal{S}_{n,k} \quad V_{\tau}\Pi_{\mathbf{i}}V_{\tau^{-1}}=\Pi_{\tau(\mathbf{i})}$ it is enough to work with the index $\mathbf{i}_0$ only. Having that and border solution for $x_{\alpha}$ from~\eqref{border}, we calculate the quantity $p^*$ from~\eqref{primal}:
\be
\begin{split}
\label{dualpp}
p^*=\frac{1}{d^{N+k}}\sum_{\mathbf{i}\in\mathcal{I}}\tr\left(\sum_{\alpha \vdash N-k}x_{\alpha}P_{\alpha} \right)&=\frac{k!\binom{N}{k}}{d^{N+k}}\sum_{\alpha}x_{\alpha}\tr P_{\alpha}\\&=\frac{k!\binom{N}{k}}{d^{2N}}\mathop{\operatorname{min}}\limits_{\mu\in\alpha}\frac{m_{\alpha}d_{\alpha}}{\lambda_{\mu}(\alpha)},
\end{split}
\ee
since $\tr P_{\alpha}=m_{\alpha}d_{\alpha}$. For showing optimality of $p^*$ we need to solve the dual problem from~\eqref{dual0} and~\eqref{dual}. We assume the following form of the operator $\Omega$ in~\eqref{dual0}:
\be
\label{omega}
\Omega=\sum_{\alpha\vdash N-k}x_{\mu^*}(\alpha)F_{\mu*}(\alpha),\qquad x_{\mu^*}(\alpha)=d^k\frac{1}{m_{\mu^*/\alpha}}\frac{m_{\alpha}}{m_{\mu^*}}.
\ee
The symbol $\mu^*$ means that we are looking for such $\mu\in\alpha$ which minimizes the quantity $p_*$ from~\eqref{dual0}. Operators $F_{\mu^*}(\alpha)$ are eigenprojectors of $\rho$ given through Definition~\ref{efy} and Theorem~\ref{eig_dec_rho}, symbol $m_{\mu^*/\alpha}$ denotes number of paths on reduced Young's lattice in which diagram $\mu^*$ can be obtained from diagram $\alpha$. Finally $m_{\mu^*},m_{\alpha}$ denote respective multiplicities of irreps. Since we are looking for any feasible solution to bound exact average probability of success $p$ from the below we are allowed for such kind of assumptions. The first constraint from~\eqref{dual} is automatically satisfied due to assumed form of $\Omega$ in~\eqref{omega}. To check the second condition we need to compute
\be
\tr_{(2k)}\left(P^+_{A_{\mathbf{i}}C}\Omega \right)= \tr_{(2k)}\left(P^+_{A_{\mathbf{i}_0}C}\Omega \right)=\frac{1}{d^k}\tr_{(2k)}\left(V^{(k)}\Omega \right),
\ee
where we used covariance property of $P^+_{A_{\mathbf{i}}C}$ and covariance of $\Omega$ with respect to the elements from the coset $\mathcal{S}_{n,k}$. Writing explicitly $\Omega$ and using Lemma~\ref{A1} we have
\be
\begin{split}
	\frac{1}{d^k}\tr_{(2k)}\left(V^{(k)}\Omega \right)&=\sum_{\alpha}\frac{1}{m_{\mu^*/\alpha}}\frac{m_{\alpha}}{m_{\mu^*}}\tr_{(2k)}\left(V^{(k)}F_{\mu^*}(\alpha) \right)\\ 
	                                                  &=\sum_{\alpha}P_{\alpha}=\mathbf{1},
\end{split}
\ee
so we satisfy the second constraint from~\eqref{dual} with equality. Now we are in position to compute $p_*$ from~\eqref{dual}:
\be
\begin{split}
\label{dualp}
p_*=\frac{1}{d^{N+k}}\tr\Omega&=\frac{1}{d^N}\sum_{\alpha}\frac{1}{m_{\mu^*/\alpha}}\frac{m_{\alpha}}{m_{\mu^*}}\tr\left(F_{\mu^*}(\alpha) \right)\\
															&=\frac{1}{d^N}\sum_{\alpha}\frac{m_{\alpha}^2d_{\mu^*}}{m_{\mu^*}}=\frac{k!\binom{N}{k}}{d^{2N}}\mathop{\operatorname{min}}\limits_{\mu\in\alpha}\frac{m_{\alpha}d_{\alpha}}{\lambda_{\mu}(\alpha)}.
															\end{split}
\ee
In third equality we use Lemma~\ref{A3}, in fourth we used the definition of the symbol $\mu^*$ and form of $\lambda_{\mu}(\alpha)$ from~\eqref{rhodec0a}. From expressions~\eqref{dualpp} and~\eqref{dualp} we see that $p^*=p_*$. We conclude that exact value of the average success probability indeed is given through expression~\eqref{exact} with corresponding measurements~\eqref{ex_measurements} presented in Theorem~\ref{thm_p}.
\end{proof}
\section{Discussion}
\label{diss}
 In this paper, we deliver analysis of the the multi-port based teleportation schemes, which are non-trivial generalisation of the famous port-based teleportation protocol. These schemes allow for teleporting several unknown quantum states (or a composite quantum state) in one go so that the states end up in the respective number of ports on Bob's side. This protocol offers much better performance than the original PBT at the price of requiring corrections on the receiver's side which are permutations of the ports where the teleported states arrive.
We discuss the deterministic protocol where the transmission always happens, but the teleported state is distorted, and the probabilistic case, where we have to accept the probability of failure, but whenever the protocol succeeds the teleportation is perfect. In both cases, we calculate parameters describing the performance of discussed schemes, like entanglement fidelity (see Theorem~\ref{Fthm}) and the probability of success (see Theorem~\ref{thm_p}). Expressions, except the global parameters such as the number of ports $N$ and local dimension $d$, depend on purely group-theoretical quantities like for example dimensions and multiplicities of irreducible representations of the permutation group. The whole analysis is possible due to the rigorous description of the algebra of partially transposed permutation operators provided in this paper. In particular, we deliver the matrix operator basis in irreducible spaces on which respective operators describing teleportation protocol are supported (see Theorem~\ref{kPBTmat}, Theorem~\ref{eig_dec_rho}). The developed formalism applied to the considered problem allows to reduce calculations from the natural representation space to every irreducible block separately, simplifying it significantly. Moreover, symmetries occurring in the protocol allow us to solve semidefinite programming problems in an analytical way, which is not granted in general in SDP problems, see Section~\ref{probkPBT}. 

The methods presented in this paper may be applied to solve some related problems, but require further development of the formalism. The first one is the construction of the optimized version of the multi-port schemes. In this case, we have to find the operation $O_A$ which Alice has to apply to her part of the resource state before she runs the protocol. Clearly in this case the resource state is no longer in the form of product of the maximally entangled pairs.  The second problem is to understand the scaling of the entanglement fidelity and probability of success in the number of ports $N$, the number of teleported particles $k$ and local dimension $d$.  To answer this question one needs to adapt the analysis presented in~\cite{majenz} and examine the asymptotic behavior of the quantity $m_{\mu/\alpha}$ appearing in our analysis (see for example Theorem~\ref{Fthm}). The third problem is to understand multi-port recycling schemes as a generalization of ideas in~\cite{Strelchuk}. We would like to know how much the resource state degrades after the teleportation procedure and is there, in principle, the possibility of exploiting the resource state again. 
\section*{Acknowledgements}
 MS, MM are supported through grant Sonatina 2, UMO-2018/28/C/ST2/00004 from the Polish National Science Centre. 
Moreover, MH and MM thank the Foundation for Polish
Science through IRAP project co-financed by the EU within
the Smart Growth Operational Programme (contract no.
2018/MAB/5). 
M.H. also acknowledges support from the National Science Centre, Poland, through grant OPUS 9, 2015/17/B/ST2/01945.
MM and PK would like to thank ICTQT Centre (University of Gda{\'n}sk) for hospitality where part of this work has been done.
This paper was presented in part at QIP 2021, Munich in a talk "Multi-port teleportation schemes" by Michał Studziński and Piotr Kopszak.

\appendix
\section{An Alternative Proof of Theorem~\ref{Fthm}}
\label{AppA}
	Using spectral decomposition of the operator $\rho$ presented in Theorem~\ref{eig_dec_rho} we expand equation~\eqref{ef} to:
	\begin{IEEEeqnarray}{lll}
		F&=\frac{k!}{d^{2k}}\binom{N}{k}\tr\left(\Pi_{\mathbf{i}_0}^{A\widetilde{B}}\sigma_{\mathbf{i}_0}^{A\widetilde{B}} \right)&\nonumber\\
		 &=\frac{k!}{d^{2N}}\binom{N}{k}\sum_{\alpha,\beta \vdash N-k}\sum_{\mu\in\alpha}\sum_{\nu\in\beta}&\frac{1}{\sqrt{\lambda_{\mu}(\alpha)}}\frac{1}{\sqrt{\lambda_{\nu}(\beta)}}\label{fa1}\\
		 &&\times\tr\left(F_{\mu}(\alpha)V^{(k)}F_{\nu}(\beta)V^{(k)} \right).\nonumber
	\end{IEEEeqnarray}
	Now we have to compute the trace from the composition $F_{\mu}(\alpha)V^{(k)}F_{\nu}(\beta)V^{(k)}$ between partially transposed permutation operator defined in~\eqref{parV} and eigenprojectors $F_{\mu}(\alpha)$ presented in Definition~\ref{efy}. Numbers $\lambda_{\mu}(\alpha)$ denote respective eigenvalues of multi-port teleportation operator given in~\eqref{rhodec0a}. Using explicit form of eigenprojectors from Definition~\ref{efy}, expression~\eqref{tmbas2} for basis operator in Theorem~\ref{tmbas}, together with  Fact~\ref{kpartr} and Lemma~\ref{L3} we can write the following chain of equalities, displayed at the top of the following page, where the simplified form~\eqref{f1} follows from $\left[E_{l_{\alpha}1_{\alpha}}^{\alpha},V^{(k)} \right]=0$ and $E_{l_{\alpha}1_{\alpha}}^{\alpha}E_{1_{\alpha} \quad l_{\alpha}}^{r_{\mu/\alpha} r_{\mu/\alpha}}=E_{l_{\alpha} \quad l_{\alpha}}^{r_{\mu/\alpha} r_{\mu/\alpha}}$, where we applied definition of $P_{\alpha}=\sum_{l_{\alpha}}E^{\alpha}_{l_{\alpha}l_{\alpha}}$, orthogonality relation $P_{\alpha}P_{\beta}=\delta^{\alpha\beta}P_{\alpha}$ and finally $\tr P_{\alpha}=m_{\alpha}d_{\alpha}$. The symbol $m_{\mu/\alpha}$ denotes number of paths on reduced Young's lattice in which frame $\mu$ can be obtained from frame $\alpha$ by adding $k$ boxes.
\begin{figure*}[!t]
	\normalsize
	\setcounter{MYtempeqncnt}{\value{equation}}
	\setcounter{equation}{150}
	\be
	\label{chain}
	\begin{split}
	&\tr\left(V^{(k)}F_{\mu}(\alpha) V^{(k)}F_{\nu}(\beta)\right) =\sum_{r_{\mu/\alpha}, r_{\nu/\beta}}\sum_{k_{\mu},j_{\nu}}\tr\left(V^{(k)}F^{r_{\mu/\alpha}r_{\mu/\alpha}}_{k_{\mu} \quad k_{\mu}}V^{(k)}F^{r_{\nu/\beta}r_{\nu/\beta}}_{j_{\nu} \quad j_{\nu}}\right)\\
	&=\sum_{r_{\mu/\alpha}, r_{\nu/\beta}}\sum_{k_{\mu},j_{\nu}}\frac{m_{\alpha}m_{\beta}}{m_{\mu}m_{\nu}}\tr\left(V^{(k)}E_{k_{\mu} \ 1_{\alpha}}^{ \quad  r_{\mu/\alpha}}V^{(k)}E_{1_{\alpha} \ k_{\mu}}^{ \quad  r_{\mu/\alpha}}V^{(k)}E_{j_{\nu}1_{\beta}}^{ \ r_{\nu/\beta}}V^{(k)} E_{1_{\beta}j_{\nu}}^{ \ r_{\nu/\beta}}\right)\\
	&=\sum_{r_{\mu/\alpha}, r_{\nu/\beta}}\sum_{r_{\mu/\alpha'}r_{\nu/\beta'}}\sum_{l_{\alpha'},s_{\beta'}}\frac{m_{\alpha}m_{\beta}}{m_{\mu}m_{\nu}}\tr\left(V^{(k)}E_{l_{\alpha'} \quad 1_{\alpha}}^{ r_{\mu/\alpha'}  r_{\mu/\alpha}}V^{(k)}E_{1_{\alpha} \quad l_{\alpha'}}^{r_{\mu/\alpha} r_{\mu/\alpha'}}V^{(k)}E_{s_{\beta'} \quad 1_{\beta}}^{r_{\nu/\beta'} r_{\nu/\beta}}V^{(k)} E_{1_{\beta} \quad s_{\beta'}}^{r_{\nu/\beta}r_{\nu/\beta'}}\right)\\
	&=\sum_{r_{\mu/\alpha}, r_{\nu/\beta}}\sum_{r_{\mu/\alpha'}r_{\nu/\beta'}}\sum_{l_{\alpha'},s_{\beta'}}\delta^{r_{\mu/\alpha}r_{\mu/\alpha'}}\delta^{r_{\nu/\beta}r_{\nu/\beta'}}\tr\left(E_{l_{\alpha}1_{\alpha}}^{\alpha}V^{(k)}E_{1_{\alpha} \quad l_{\alpha}}^{r_{\mu/\alpha} r_{\mu/\alpha}}E_{s_{\beta}1_{\beta}}^{\beta}V^{(k)}E_{1_{\beta} \quad s_{\beta}}^{r_{\nu/\beta}r_{\nu/\beta}} \right)
	\end{split}
	\ee
	which simplifies to
\be
\label{f1}
\begin{split}
&\tr\left(V^{(k)}F_{\mu}(\alpha) V^{(k)}F_{\nu}(\beta)\right) =\sum_{r_{\mu/\alpha}, r_{\nu/\beta}}\sum_{l_{\alpha},s_{\beta}}\tr\left(V^{(k)}E_{l_{\alpha} \quad l_{\alpha}}^{r_{\mu/\alpha} r_{\mu/\alpha}}V^{(k)}E_{s_{\beta} \quad s_{\beta}}^{r_{\nu/\beta}r_{\nu/\beta}} \right)\\
&=\sum_{r_{\mu/\alpha}, r_{\nu/\beta}}\sum_{l_{\alpha},s_{\beta}}\frac{m_{\mu}}{m_{\alpha}}\tr\left(E^{\alpha}_{l_{\alpha}l_{\alpha}}V^{(k)}E_{s_{\beta} \quad s_{\beta}}^{r_{\nu/\beta}r_{\nu/\beta}} \right)=\sum_{r_{\mu/\alpha}, r_{\nu/\beta}}\sum_{k_{\alpha},j_{\beta}}\frac{m_{\mu}}{m_{\alpha}}\tr\left(E^{\alpha}_{k_{\alpha}k_{\alpha}}E_{j_{\beta} \quad j_{\beta}}^{r_{\nu/\beta}r_{\nu/\beta}} \right)\\
&=\sum_{r_{\mu/\alpha}, r_{\nu/\beta}}\sum_{l_{\alpha},s_{\beta}}\frac{m_{\mu}}{m_{\alpha}}\frac{m_{\nu}}{m_{\beta}}\tr\left(E^{\alpha}_{l_{\alpha}l_{\alpha}}E^{\beta}_{s_{\beta}s_{\beta}}\right)=\frac{m_{\mu}}{m_{\alpha}}\frac{m_{\nu}}{m_{\beta}}m_{\mu/\alpha}m_{\nu/\beta}\tr(P_{\alpha}P_{\beta})\\
&=\delta^{\alpha\beta}\frac{m_{\mu}m_{\nu}}{m_{\alpha}^2}m_{\mu/\alpha}m_{\nu/\beta}\tr P_{\alpha}=\delta^{\alpha\beta}m_{\mu/\alpha}m_{\nu/\alpha}m_{\mu}m_{\nu}\frac{d_{\alpha}}{m_{\alpha}}.
\end{split}
\ee

\setcounter{equation}{\value{MYtempeqncnt}+2}
\hrulefill
\vspace*{4pt}
\end{figure*}
Substituting final form of~\eqref{f1} to~\eqref{fa1} we have
	\be
	\begin{split}
		F&=\frac{k!}{d^{2N+2k}}\binom{N}{k}\sum_{\alpha \vdash N-k}\sum_{\mu,\nu\in\alpha}\frac{d_{\alpha}}{m_{\alpha}}\frac{m_{\mu}}{\sqrt{\lambda_{\mu}(\alpha)}}\frac{m_{\nu}}{\sqrt{\lambda_{\nu}(\beta)}}m_{\mu/\alpha}m_{\nu/\alpha}.
	\end{split}
	\ee
	
	Inserting explicit form of eigenvalues $\lambda_{\nu}(\alpha),\lambda_{\mu}(\alpha)$ given in~\eqref{rhodec0a}, we reduce to:
	\begin{IEEEeqnarray}{rll}
		F&=\frac{k!}{d^{2N+2k}}\binom{N}{k}\sum_{\alpha \vdash N-k}&\sum_{\mu,\nu\in\alpha}\frac{d_{\alpha}}{m_{\alpha}}\frac{m_{\mu}}{\sqrt{\lambda_{\mu}(\alpha)}}\frac{m_{\nu}}{\sqrt{\lambda_{\nu}(\beta)}}m_{\mu/\alpha}m_{\nu/\alpha}\nonumber\\
		 &=\frac{k!}{d^{2N+2k}}\binom{N}{k}\frac{d^{N}}{k!\binom{N}{k}}&\sum_{\alpha \vdash N-k}\sum_{\mu,\nu\in\alpha}m_{\mu/\alpha}m_{\nu/\alpha}\nonumber\\ 
		 &&\sqrt{\frac{m_{\alpha}d_{\mu}}{m_{\mu}d_{\alpha}}}\sqrt{\frac{m_{\alpha}d_{\nu}}{m_{\nu}d_{\alpha}}}\frac{m_{\mu}m_{\nu}}{m_{\alpha}}d_{\alpha}\nonumber\\
		 &=\frac{1}{d^{N+2k}}\sum_{\alpha \vdash N-k}\sum_{\mu,\nu\in\alpha}&m_{\mu/\alpha}\sqrt{m_{\mu}d_{\mu}}m_{\nu/\alpha}\sqrt{m_{\nu}d_{\nu}}\nonumber\\
		 &=\frac{1}{d^{N+2k}}\sum_{\alpha \vdash N-k}\bigg(\sum_{\mu\in\alpha}&m_{\mu/\alpha} \sqrt{m_{\mu}d_{\mu}}\bigg)^2.
\end{IEEEeqnarray}
 This finishes the proof.
 \bibliographystyle{IEEEtran}
 \bibliography{IEEEabrv,biblio2}
 \begin{IEEEbiographynophoto}{Micha{\l} Studzi{\'n}ski} received MSc degree in astronomy in 2009 from the Nicolaus Copernicus University and PhD in physics from the University of Gda{\'n}sk in 2015. Next, for a one year he was a postdoc at the National Quantum Information Centre, Sopot, Poland.  Then he has spend three years (2016-2018) as a postdoc at the Department of Applied Mathematics and Theoretical Physics, the University of Cambridge, United Kingdom. From 2019 to 2021 he has been working as a researcher at the Faculty of Physics, Mathematics and Informatics, University of Gda{\'n}sk, Gda{\'n}sk, Poland. 
From 2022 he is an adjunct at the Institute of Theoretical Physics and Astrophysics, University of Gda{\'n}sk, Gda{\'n}sk, Poland. 

He works in quantum information theory and mathematical physics. One of the most important results co-authored by him are fundamental limitations on coherence transfer under thermal operation and group-theoretic description of the  port-based teleportation protocol.
 \end{IEEEbiographynophoto}
 \begin{IEEEbiographynophoto}{Marek Mozrzymas} was born in Wrocław, Poland in 1960. He received MSc degree (1985) and PhD degree (1991) in theoretical physics from the University of Wrocław in Poland and Habilitation degree in theoretical physics in 2005 also from the University of Wrocław.  Since 1985 he has been  an academic in the Institute of Theoretical Physics at the University of Wrocław, and since 2017 he has been  professor at the University of Wrocław. He is head of PhD Studies in the Faculty of Physics and Astronomy at University of Wrocław. He is also a member of the Scientific Council of the National Quantum Information Centre in Poland.
His research interests include applications of algebraic methods in theoretical physics, in particular the applications of quantum algebras and representation theory of groups and semisimple algebras in physics and recently in quantum information protocols.
 \end{IEEEbiographynophoto}
 \begin{IEEEbiographynophoto}{Piotr Kopszak}
	 received a MSc degree in mathematics in 2015 from Wroc{\l}aw University of Tehnology (specialization in mathematical statistiscs) and in physics in 2017 from University of Wroc{\l}aw (specialization in theoretical physics). Later he become a PhD student at the Institute of Theorethical Physics at the Wroc{\l}aw University. His research coveres application of representation theory in quantum physics. In particular it involved providing new examples of of positive maps (entanglement witnesses) as well as the study of the port-based teleportation protocols.
 \end{IEEEbiographynophoto}
 \begin{IEEEbiographynophoto}{Micha{\l} Horodecki}
received the MSc degree in physics in 1996 and  the PhD degree in theoretical physics from the University of Gdansk in 2000. In 2015 he has become a full professor at University of Gdańsk, and since 2019 is group leader in newly established International Centre for Theory of Quantum Technologies.  

He works in quantum information theory as well as in quantum open systems. His most recognized achievements are co-discovering bound entanglement, and quantum state merging. He also worked on random quantum circuits and quantum cryptography. Recently he was involved in quantum thermodynamics, developing the so-called resource theory of thermodynamics. His present interests include quantum thermal machines and open quantum systems. 
 \end{IEEEbiographynophoto}
\end{document}